\documentclass[]{interact}

\usepackage{epstopdf}
\usepackage[caption=false]{subfig} 

\usepackage{color}
\usepackage{algorithm,algcompatible,amsmath}
\usepackage[natbibapa,nodoi]{apacite}
\theoremstyle{plain}
\newtheorem{theorem}{Theorem}[section]
\newtheorem{them}[theorem]{Theorem}
\newtheorem{lemma}[theorem]{Lemma}
\newtheorem{coro}[theorem]{Corollary}
\newtheorem{prop}[theorem]{Proposition}

\theoremstyle{definition}

\theoremstyle{remark}

\begin{document}

\articletype{ARTICLE TEMPLATE}

\title{A Class of Distributed Event-Triggered Average Consensus Algorithms for Multi-Agent Systems}

\author{
\name{Ping Xu,\quad Cameron Nowzari,\quad Zhi Tian}
\affil{Department of Electrical and Computer Engineering, George Mason University, Fairfax, VA, 22030, USA}
}


\maketitle

\begin{abstract}
This paper proposes a class of distributed event-triggered algorithms that solve the average consensus problem in multi-agent systems. By designing events such that a specifically chosen Lyapunov function is monotonically decreasing, event-triggered algorithms succeed in reducing communications among agents while still ensuring that the entire system converges to the desired state. However, depending on the chosen Lyapunov function the transient behaviors can be very different. Moreover, performance requirements also vary from application to application. Consequently, we are instead interested in considering a class of Lyapunov functions such that each Lyapunov function produces a different event-triggered coordination algorithm to solve the multi-agent average consensus problem. The proposed class of algorithms all guarantee exponential convergence of the resulting system and exclusion of Zeno behaviors. This allows us to easily implement different algorithms that all guarantee correctness to meet varying performance needs. We show that our findings can be applied to the practical clock synchronization problem in wireless sensor networks (WSNs) and further corroborate their effectiveness with simulation results.   
 
 
\end{abstract}
 

\section{Introduction}
\label{sec:intro}

The consensus problem of multi-agent systems where a group of agents are required to agree upon certain quantities of interest finds broad applications in areas such as unmanned vehicles, mobile robots, and wireless sensor networks (WSNs)~\citep{liang2012distributed, peng2015distributed, olfati2012coupled}. Toward this problem, one effective and efficient method is the distributed event-triggered coordination approach, which was first proposed in~\citep{dimarogonas2009event}, and have been studied extensively over the last decades~\citep{xie2015event, yi2016distributed, nowzari2016distributed, liu2018fixed}, see references in~\citep{ding2017overview, nowzari2019event} for recent advances and more details.   
 

The main idea behind distributed event-triggered algorithms is that the iterative communication between agents and their one-hop neighbors only happens when certain conditions/events are triggered. Through skipping unnecessary communications, the communication efficiency is increased, and at the same time the desired properties of the system are maintained. The triggering conditions of the event-triggered algorithms can be time-dependent~\citep{seyboth2013event}, state-dependent~\citep{nowzariZeno-free_2014, nowzari2016distributed, liu2018fixed}, or a combination of both~\citep{girard2015dynamic, sun2016new, yi2017distributed}. In general, the time-dependent thresholds are easy to design to exclude deadlocks (or Zeno behavior, meaning an infinite number of events triggered in a finite number of time period~\citep{johansson1999regularization}), but require global information to guarantee convergence to exactly a consensus state. While state-dependent thresholds are easier to design, these triggers might be risky to implement as Zeno behavior is harder to exclude. As the occurrence of Zeno behavior is impossible in a given physical implementation, the exclusion of it is therefore necessary and essential to guarantee the correctness of an event-triggered algorithm.

In this paper, we focus on developing event-triggered algorithms with state-dependent triggering thresholds that exclude the Zeno behavior. To be specific, an event-triggered controller with state-dependent triggering thresholds can generally be developed from a given Lyapunov function to maintain stability of a certain system while reducing sampling or communication, using the given Lyapunov function as a certificate of correctness. In other words, all events are triggered based on how we want the given Lyapunov function to evolve in time. However, there are no formal guarantees on the gained efficiency. Moreover, it is known that a Lyapunov function is not unique for a given system, and each individual function may result in a totally different, but equally valid/correct triggering law. Consequently, there are many works that propose one such algorithm based on one function that all have the same guarantee: asymptotic convergence to a consensus state. That means there is no established way to compare the performance of two different event-triggered algorithms that solve the same problem. In particular, given two different event-triggered algorithms that both guarantee convergence, their trajectories and communication schedules may be wildly different before ultimately converging to the desired set of states. There are some new works that are addressing exactly this topic~\citep{ramesh2016performance, khashooei2017output, borgers2017tradeoffs, heijmans2017stability}, which set the basis for this paper. More specifically, once established methods of comparing the performance of event-triggered algorithms against one another are developed, current available algorithms will likely be revisited to optimize different types of performance metrics. In particular, we notice that different algorithms are better than others in different scenarios when considering metrics such as convergence speed or total energy consumption. Therefore, instead of trying to design only one event-triggered algorithm that simply guarantees convergence, we design an entire class of event-triggered algorithms that can be easily tuned to meet varying performance needs.


Our work is motivated by \citep{nowzari2016distributed} that solves the exact problem we consider, i.e., design a distributed event-triggered algorithm with state-dependent triggers for multi-agent systems over weight-balanced directed graphs. We first develop a distributed event-triggered algorithm based on an alternative Lyapunov candidate function, which we name it as \textbf{Algorithm 2}. For the algorithm proposed by~\citet{nowzari2016distributed}, we name it as \textbf{Algorithm 1}. Observing that the two algorithms result in different performance for different network topologies, we then parameterize an entire class of Lyapunov functions from the two algorithms and show how each individual function can be used to develop a \textbf{Combined Algorithm}. More specifically, choosing any parameter~$\lambda \in [0,1]$ yields an event-triggered algorithm that guarantees convergence. Changing~$\lambda$ can then help achieve varying performance goals while always guaranteeing stability. With the asymptotic convergence and exclusion of Zeno behavior for both \textbf{Algorithm 1} and \textbf{Algorithm 2}, we establish that the entire class of \textbf{Combined Algorithms} also exclude Zeno behavior and guarantee convergence of the system. In addition to the theoretic analysis, we also study the practical clock synchronization problem that exists in WSNs~\citep{dimarogonas2009event}, which is crucial especially when operations such as data fusion, power management and transmission scheduling are performed~\citep{wu2011clock, kadowaki2015event}. We use various simulations to illustrate the correctness and performance of our proposed algorithms. 
 
The rest of this paper is organized as follows. Section~\ref{sec:preli} introduces the preliminaries and Section~\ref{sec:probstate} formulates the problem of interest. Section~\ref{sec:triggerdesign} first summarizes the related work~\citep{nowzari2016distributed} and then proposes a novel strategy based on an alternative Lyapunov function. Section~\ref{sec:analysis} analyzes the non-Zeno behavior and convergence property of the proposed strategy. The combined algorithms that are developed based on the combined Lyapunov functions are proposed in Section~\ref{sec:combined}, followed by a case study of clock synchronization in Section~\ref{sec:clock_sync}. Section~\ref{sec:sim} presents the simulation results and Section~\ref{sec:con} concludes this work.

\textbf{Notations:} $\mathbb{R},\;\mathbb{R}_{>0},\;\mathbb{R}_{\geq 0}$ denote the set of real, positive real, and nonnegative real numbers, respectively. $\mathbf{1}_N\in \mathbb{R}^N$ and $\mathbf{0}_N\in \mathbb{R}^N$ denote the $N\times 1$ column vectors with entries all equal to one and zero, respectively. $\|\cdot\|$ denotes the Euclidean norm for vectors or induced 2-norm for matrices. For a finite set $S$, $|S|$ denotes its cardinality. 

\section{Preliminaries}
\label{sec:preli}
Let $\mathcal{G}=\{\mathcal{V,E},W\}$ denote a weighted directed graph (or weighted digraph) that is comprised of a set of vertices $\mathcal{V}=\{1,\dots,N\}$, directed edges $\mathcal{E}\subset\mathcal{V}\times\mathcal{V}$, and weighted adjacency matrix $W \in \mathbb{R}_{\geq 0}^{N\times N}$. Given an edge $(i,j)\in \mathcal{E}$, we refer to $j$ as an out-neighbor of $i$ and $i$ as an in-neighbor of $j$. The sets of out- and in-neighbors of a given agent $i$ are $\mathcal{N}_i^{out}$ and $\mathcal{N}_i^{in}$, respectively. The weighted adjacency matrix $W$ satisfies $w_{ij}>0$ if $(i,j)\in \mathcal{E}$ and $w_{ij}=0$ otherwise. A path from vertex $i$ to $j$ is an ordered sequence of vertices such that each intermediate pair of vertices is an edge. A digraph $\mathcal{G}$ is strongly connected if there exists a path from all $i\in \mathcal{V}$ to all $j\in \mathcal{V}$. The out- and in-degree matrices $D^{out}$ and $D^{in}$ are diagonal matrices whose diagonal elements are
\begin{equation}
\textstyle d_i^{out}=\sum_{j\in \mathcal{N}_i^{out}} w_{ij}, \quad d_i^{in}=\sum_{j\in \mathcal{N}_i^{in}} w_{ji},
\nonumber
\end{equation}
respectively. A digraph is weight-balanced if $D^{out}=D^{in}$, and the weighted Laplacian matrix is given by $L=D^{out}-W$.

For a strongly connected and weight-balanced digraph, zero is a simple eigenvalue of $L$. In this case, we order its eigenvalues as $\lambda_1=0<\lambda_2\leq \dots \leq \lambda_N$. Note the following property will be of use later:
\begin{equation}
\label{eq:eigen}
\textstyle \lambda_2(L)x^TL^Tx\leq x^TL^TLx \leq \lambda_N(L)x^TL^Tx.
\end{equation}

Another property we need is the Young's inequality~\citep{hardy1952inequalities}, which states that given $x,y\in\mathbb{R}$, for any $\varepsilon\in \mathbb{R}_{>0}$,
\begin{equation}
\label{eq:youngs}
\textstyle xy\leq \frac{x^2}{2\varepsilon}+\frac{\varepsilon y^2}{2}.
\end{equation}
 
\section{Problem Statement}
\label{sec:probstate}
Consider the average consensus problem for an $N$-agent network described by a weight-balanced and strongly connected digraph $\mathcal{G}=\{\mathcal{V,E},W\}$. Without loss of generality, we say that an agent $i$ is able to receive information from neighbors in~$\mathcal{N}_i^{out}$ and send information to neighbors in~$\mathcal{N}_i^{in}$. Assume that all inter-agent communications are instantaneous and of infinite precision. Let~$x_i$~denote the state of agent~$i\in\mathcal{V}$~and consider the single-integrator dynamics
\begin{equation}
\label{eq:single_dymic}
\textstyle \dot{x}_i(t)=u_i(t).
\end{equation}

The well-known distributed continuous control law
\begin{equation}
\label{eq:ctr_law}
\textstyle u_i(t)=-\sum_{j\in\mathcal{N}_i^{out}}w_{ij}(x_i(t)-x_j(t))
\end{equation}
drives the states of all agents in the system to asymptotically converge to the average of their initial states~\citep{olfati2004consensus}. However, its implementation requires all agents to continuously access their neighbors' state information and keep updating their own control signals, which is practically unrealistic in terms of both communication and control. To relax both of these requirements, we adopt the modified distributed event-triggered control law~\citep{dimarogonas2012distributed}
\begin{equation}
\label{eq:ctr_law_modified}
\textstyle u_i(t)=-\sum_{j\in\mathcal{N}_i^{out}}w_{ij}(\hat{x}_i(t)-\hat{x}_j(t)),
\end{equation}
where $\hat{x}_i(t)$ to denote the last broadcast state of agent $i$ and it remains constant between two broadcasts. That is, if we let $t_{last}$ be the last time at which agent $i$ broadcasts its state information and $t_{next}$ be the next time it is going to broadcast, then $\hat{x}_i(t)=x_i({t_{last}})$ for $t\in [t_{{last}},t_{next})$. With this framework, neighbors of a given agent are able to receive state information from it only when this agent decides to broadcast its state information to them. After receiving the information from their neighbors, agents then update their own control signals.

Along with the above controller~\eqref{eq:ctr_law_modified}, each agent $i$ is equipped with a triggering function $f_i(\cdot)$ that takes values in $\mathbb{R}$. Our first objective is to identify triggers that depend on local information only, i.e., on the true state $x_i(t)$, its last broadcast state $\hat{x}_i(t)$, and its neighbors last broadcast state $\hat{x}_j(t)$ for $j\in \mathcal{N}_i$. Specifically, we need to design triggering functions for each agent $i\in \mathcal{V}$ such that an event is triggered as soon as the triggering condition
\begin{equation}
\label{eq:triggering_condition}
\textstyle f_i(t,x_i(t),\hat{x}_i(t),\hat{x}_j(t))> 0
\end{equation}
is fulfilled. The triggered event then drives agent $i$ to broadcast its state so that its neighbors can update their states. 
To do so, the general steps are to identify a Lyapunov function for the system, and then derive triggering rules from the Lyapunov function while maintaining the stability of the system and ensures asymptotic convergence to a consensus state.
	
Notice that a Lyapunov function is not unique for a given system, and each individual function may result in a totally different, but equally valid/correct triggering law. Moreover, when considering metrics such as convergence speed or total energy consumption, different algorithms are better than others in different scenarios. Since there is no established way to compare the performance of two different event-triggered algorithms that solve the same problem and performance requirements may vary from application to application, therefore, our second objective is to design an entire class of event-triggered algorithms that can be easily tuned to meet varying performance needs. Before presenting our work, we first introduce the algorithm that motivates our work~\citep{nowzari2016distributed}. 
 

\section{Distributed Trigger Design}
\label{sec:triggerdesign}

\subsection{Related work}
\label{subsec:related}
The exact same problem of distributed event-triggered coordination for multi-agent systems over weight-balanced digraphs has been studied by~\citet{nowzari2016distributed}. As their findings are essential in developing our algorithms, we first summarize their algorithm and name it \textbf{Algorithm 1}.
 
The event-triggered law proposed in~\citep{nowzari2016distributed} is Lyapunov-based, with the Lyapunov candidate function be
\begin{equation}
\label{eq:Lyap_V1}
\textstyle V_1(x(t)) = \frac{1}{2}(x(t)-\bar{x})^T(x(t)-\bar{x}),
\end{equation}
where $x(t)=(x_1(t),...,x_N(t))^T\in \mathbb{R}^N$ is the column vector of all
agents' states and $\bar{x}=\frac{1}{N}\sum_{i=1}^Nx_i(0)\mathbf{1}_N$ is the average of all initial conditions.

The derivative of $V_1(x(t))$ takes the form
\begin{equation}
\label{eq:Lyap_V1_dot}
\textstyle
\begin{split}
\dot{V}_1(x(t)) &= x^T(t)\dot{x}(t)-\bar{x}^T\dot{x}(t) =-x^T(t)L\hat{x}(t) + \bar{x}^TL\hat{x}(t)  = -x^T(t)L\hat{x}(t),
\end{split}
\end{equation}
where $\dot{x}(t)=u(t)=-L\hat{x}(t)$ is the compact vector-matrix form of equation \eqref{eq:single_dymic} and \eqref{eq:ctr_law_modified}, with $\hat{x}(t)=(\hat{x}_1(t),...,\hat{x}_N(t))^T \in \mathbb{R}^N$ the vector of last broadcast states of all agents. The second term $\bar{x}^TL\hat{x}(t)=0$ comes from the fact that the digraph $\mathcal{G}$ is weight-balanced, meaning $\mathbf{1}_N^TL=\mathbf{0}^T$, therefore $\bar{x}^TL\hat{x}(t)=\frac{1}{N}\sum_{i=1}^Nx_i(0)\mathbf{1}_N^TL\hat{x}(t)=0$.

Expand \eqref{eq:Lyap_V1_dot} and apply Young's inequality~\eqref{eq:youngs}, $\dot{V}_1(x(t))$ is upper bounded by
\begin{equation}
\label{eq:V1_dot_upbound}
\textstyle \dot{V}_1(x(t))\leq -\frac{1}{2}\sum_{i=1}^N \sum_{j\in \mathcal{N}_i^{out}} w_{ij} \Big[(1-a_i)(\hat{x}_i(t)-\hat{x}_j(t))^2-\frac{e_i^2(t)}{a_i} \Big],
\end{equation}
where $a_i\in (0,1)$ and $e_i(t)=\hat{x}_i(t)-x_i(t)$ is the difference between agent $i$'s last broadcast state and its current state at time $t$.

To make sure that the Lyapunov function $V_1(x(t))$ is monotonically decreasing requires
\begin{equation}
\textstyle \sum_{j\in \mathcal{N}_i^{out}} w_{ij} \Big[(1-a_i)(\hat{x}_i(t)-\hat{x}_j(t))^2-\frac{e_i^2(t)}{a_i} \Big]\geq 0,
\nonumber
\end{equation}
for all agents $i \in \mathcal{V}$ at all times, which can be accomplished by enforcing
\begin{equation}
\label{eq:Alg1_error_bound}
\textstyle e_i^2(t)\leq \frac{a_i(1-a_i)}{d_i^{out}}\sum_{j\in \mathcal{N}_i^{out}} w_{ij}(\hat{x}_i(t)-\hat{x}_j(t))^2.
\end{equation}

It is found in \citep{nowzari2016distributed} that by setting $a_i=0.5$ for all agents, the trigger design will be optimal. Therefore, the triggering function in \citep{nowzari2016distributed} is defined as
\begin{equation}
\label{eq:Alg1_trig_func}
\textstyle f_i(e_i(t))=e_i^2(t)-\frac{\sigma_i}{4d_i^{out}}\sum_{j\in \mathcal{N}_i^{out}}w_{ij}(\hat{x}_i(t)-\hat{x}_j(t))^2,
\end{equation}
where $\sigma_i \in (0,1)$ is a design parameter that affects the flexibility of the triggers. According to the triggering function \eqref{eq:Alg1_trig_func}, an event is triggered when $f_i(e_i(t))>0$ or when $f_i(e_i(t))=0$ and $\phi_i=\sum_{j\in \mathcal{N}_i^{out}}w_{ij}(\hat{x}_i(t)-\hat{x}_j(t))^2\neq0$.

Basically, the trigger above makes sure that $\dot{V}_1(x(t))$ is always negative as long as the system has not converged, therefore, \textbf{Algorithm 1} guarantees all agents to converge to the average of their initial states, i.e., $\lim_{t\rightarrow \infty}x(t)=\bar{x}=\frac{1}{N}\sum_{i=1}^Nx_i(0)\mathbf{1}_N$,
interested readers are referred to \citep[Theorem 5.3]{nowzari2016distributed} for more details.

\subsection{Proposed new algorithm}
\label{subsec:proposed}
As we know, the Lyapunov function is not unique for the stability studying of the same system, and each individual function may result a totally different triggering law. Therefore, we propose a novel triggering strategy named \textbf{Algorithm 2} based on an alternative Lyapunov candidate function
\begin{equation}
\label{eq:Lyap_V2}
\textstyle V_2(x(t))=\frac{1}{2}x(t)^TL^Tx(t).
\end{equation}
The following result characterizes a local condition for all agents in the network such that the Lyapunov candidate function $V_2(x(t))$ is monotonically nonincreasing.
\begin{lemma}
\label{lemma:V2_dot_upbound}
For $i \in \mathcal{V}$, with $b_i,c_j<\frac{1}{d_i^{out}}\; \forall i,j \in \mathcal{V}$, define $e_i(t)=\hat{x}_i(t)-x_i(t)$ as in Section \ref{subsec:related}, with $u_i(t)$ given in \eqref{eq:ctr_law_modified}, then
\begin{equation}
\label{eq:V2_dot_upbound}
\textstyle \dot{V}_2(x(t))\leq -\sum_{i=1}^N \left[\delta_i u_i^2(t)-\left(\frac{d_i^{out}}{2b_i}+\frac{d_i^{out}}{2c_i}\right)e_i^2(t)\right),
\end{equation}
where
\begin{equation}
\label{eq:delta}
\textstyle \delta_i\triangleq 1-\frac{d_i^{out}b_i}{2}-\sum_{j\in \mathcal{N}_i^{out}}\frac{w_{ij}c_j}{2}.
\end{equation}
\end{lemma}

\begin{proof}
See Appendix A.
\end{proof}


From Lemma \ref{lemma:V2_dot_upbound}, a sufficient condition to guarantee the proposed Lyapunov candidate function $V_2(x(t))$ is monotonically decreasing is to ensure that
\begin{align*}
\textstyle \delta_i u_i^2(t)-\Big(\frac{d_i^{out}}{2b_i}+\frac{d_i^{out}}{2c_i}\Big)e_i^2(t)\geq 0
\nonumber
\end{align*}
for all agents $i \in \mathcal{V}$ at all times, or
\begin{equation}
\label{eq:Alg2_error_bound}
\textstyle e_i^2(t) \leq \frac{2\delta_ib_ic_i}{(b_i+c_i)d_i^{out}}\Big(\sum_{j\in\mathcal{N}_i^{out}}w_{ij}(\hat{x}_i(t)-\hat{x}_j(t))\Big)^2.
\end{equation}

The triggering function developed from \textbf{Algorithm 2} is therefore derived as \begin{equation}
\label{eq:Alg2_trig_func}
\textstyle f_i(e_i(t))=e_i^2(t)-\frac{2\sigma_i \delta_ib_ic_i}{(b_i+c_i)d_i^{out}}\Big(\sum_{j\in\mathcal{N}_i^{out}}w_{ij}(\hat{x}_i(t)-\hat{x}_j(t))\Big)^2,
\end{equation}
where $\sigma_i\in (0,1)$ is a design parameter that affects how flexible the trigger is and controls the trade-off between communication and performance. Setting $\sigma_i$ close to 0 is generally greedy, meaning that the trigger is enabled more frequently and more communications are required, therefore makes agent $i$ contribute more to the decrease of the Lyapunov function $V_2(x(t))$, leading to a faster convergence of the network while setting the value of $\sigma_i$ close to 1 achieves the opposite results. Note that the roles of $b_i, c_i, c_j$ are beyond system stabilization, they are also important to the trigger's performance. The larger value of $\frac{2\delta_ib_ic_i}{(b_i+c_i)d_i^{out}}$, the less communication shall be needed since it means that the system is more error-tolerant.

\begin{coro}
\label{coro:Alg2_coro}
For agent $i \in \mathcal{V}$ with the triggering function defined in \eqref{eq:Alg2_trig_func}, if the condition $f_i(e_i)\leq 0$ is enforced at all times, then
\begin{equation}
\textstyle \dot{V}_2(x(t))\leq -\sum_{i=1}^N (1-\sigma_i)\delta_i(\sum_{j\in\mathcal{N}_i^{out}}w_{ij}(\hat{x}_i(t)-\hat{x}_j(t)))^2.
\nonumber
\end{equation}
\end{coro}

Similar as the work done in \citep{nowzari2016distributed}, to avoid the possibility that agent $i$ may miss any triggers, we define an event either by
\begin{align}
 f_i(e_i(t))&>0 \quad or \label{eq:Alg2_event1}\\
  f_i(e_i(t))&=0 \quad and \quad \phi_i \neq 0\label{eq:Alg2_event2}
\end{align}
where $\phi_i=(\sum_{j\in\mathcal{N}_i^{out}}w_{ij}(\hat{x}_i(t)-\hat{x}_j(t)))^2$.

We also prescribe the following additional trigger as in \citep{nowzari2016distributed} to address the non-Zeno behavior. Let $t_{last}^i$ be the last time at which agent $i$ broadcasts its information to its neighbors. If at some time $t\geq t_{last}^i$, agent $i$ receives information from a neighbor $j\in \mathcal{N}_i^{out}$, then agent $i$ immediately broadcasts its state if
\begin{equation}
\label{eq:Alg2_time_interval}
\textstyle t \in (t_{last}^i, t_{last}^i+\varepsilon_i),
\end{equation}
where
\begin{equation}
\label{eq:Alg2_varepsilon}
\textstyle \varepsilon_i < \sqrt{\frac{2\sigma_i \delta_ib_ic_i}{(b_i+c_i)d_i^{out}}}
\end{equation}
is a parameter selected to ensure the exclusion of Zeno behavior, and we will demonstrate how it is designed in the following section.

We summarize the differences between \textbf{Algorithm 1} proposed in~\citep{nowzari2016distributed} and \textbf{Algorithm 2} proposed here in Table~\ref{tab:difference}. Once the triggering function and parameters~$\varepsilon_i$ are chosen for each agent, either algorithm can be implemented using the coordination algorithm provided in Table~\ref{tab:algorithm}. 

Note that both algorithms guarantee exponential convergence and the exclusion of Zeno behavior, as analyzed in Section \ref{sec:analysis} and in~\citep[Section 5]{nowzari2016distributed}. However, except for these similarities, we have no idea which algorithm works better for under varying performance need and initial conditions, which motivates our work in Section \ref{sec:combined}.  

\begin{table}
\centering
\tbl{Difference between \textbf{Algorithm 1} and \textbf{Algorithm 2}.}
{\begin{tabular}{|c|c|c|}
  \hline
  	&		 Triggering function & Parameter design    \\
  \hline
 \textbf{Algorithm 1} & $f_i(e_i)\triangleq e_i^2(t)-\frac{\sigma_i}{4d_i^{out}}\sum_{j\in \mathcal{N}_i^{out}}w_{ij}(\hat{x}_i(t)-\hat{x}_j(t))^2$ & $\varepsilon_i < \sqrt{\frac{\sigma_i}{4d_i^{out}w_i^{\max}|\mathcal{N}_i^{out}|}}$  \\
  \hline
  \textbf{Algorithm 2}&$f_i(e_i)\triangleq e_i^2-\frac{2\sigma_i \delta_ib_ic_i}{(b_i+c_i)d_i^{out}}\Big(\sum_{j\in\mathcal{N}_i^{out}}w_{ij}(\hat{x}_i-\hat{x}_j\Big)^2$ &$\varepsilon_i < \sqrt{\frac{2\sigma_i \delta_ib_ic_i}{(b_i+c_i)d_i^{out}}}$\\
 \hline
\end{tabular}}
\label{tab:difference}
\end{table}

\begin{table}
  \centering
  \caption{Distributed Event-Triggered Coordination Algorithm.}\label{tab:algorithm}
  \framebox[.9\linewidth]{\parbox{.85\linewidth}{%
      \parbox{\linewidth}{At all times $t$, agent $i \in \{1,\dots,N\}$ performs:}
      \vspace*{-2.5ex}
      \begin{algorithmic}[1]
        \IF{$f_i(e_i(t))>0$ or ($f_i(e_i(t))=0$ and $\phi_i \neq 0)$}
        \STATE broadcast state information $x_i(t)$ and update control signal~$u_i(t)$
        \ENDIF
        \IF{new information $x_j(t)$ is received from some neighbor(s) $j \in \mathcal{N}_i^{out}$}
        \IF{agent $i$ has broadcast its state at any time $t'\in[t-\varepsilon_i,t)$}
        \STATE broadcast state information $x_i(t)$
        \ENDIF
        \STATE update control signal $u_i(t)$
        \ENDIF
      \end{algorithmic}}}
\end{table}

\section{Stability Analysis of \textbf{Algorithm 2}}
\label{sec:analysis}
In this section, we show that \textbf{Algorithm 2} guarantees that no Zeno behavior exists in the network executions. In addition, we show that when executing \textbf{Algorithm 2}, all agents converge exponentially to the average of their initial states.

\begin{prop}
\label{prop:non-Zeno}
(Non-Zeno Behavior) Consider the system~\eqref{eq:single_dymic} executing control law~\eqref{eq:ctr_law_modified}. The triggering function is given by~\eqref{eq:Alg2_trig_func}. If the underlying digraph of the system is weight-balanced and strongly connected, then when executing the algorithm described in Table~\ref{tab:algorithm}, the system with any initial conditions will not exhibit Zeno behavior.

\begin{proof} To prove that the system does not exhibit Zeno behavior, we need to show that no agent broadcasts its state an infinite number of times in any finite time period. We divide the proof into two steps, the first step shows the existence of that finite time period and gives its value; while in the second step, we show that no information can be transmitted an infinite number of times in that finite time period.

\noindent

\textbf{Step 1}: This step shows that if an agent does not receive new information from its out-neighbors, its inter-events time is bounded by a positive constant.

Assume that agent $i\in\mathcal{V}$ has just broadcast its state at time $t_0$, then $e_i(t_0)=0$. For $t>t_0$, while no new information is received, $\hat{x}_i(t)$ and $\hat{x}_j(t)$ remain unchanged. Given that $\dot{e}_i=-\dot{x}_i$, the evolution of the error is simply
\begin{equation}
\label{eq:Alg2_error_evol}
\textstyle e_i(t)=-(t-t_0)\hat{z}_i,
\end{equation}
where $\hat{z}_i=\sum_{j\in\mathcal{N}_i^{out}}w_{ij}(\hat{x}_j-\hat{x}_i)$. Since we are considering the case that no neighbors of agent $i$ broadcast their states, therefore trigger \eqref{eq:Alg2_time_interval} is irrelevant. We then need to find out the next time point $t^*$ when $f_i(e_i(t^*))=0$ and agent $i$ is triggered to broadcast. This can be done following trigger \eqref{eq:Alg2_event2}. If $\hat{z}_i=0$, no broadcasts will ever happen because $e_i(t)=0$ for all $t\geq t_0$. Consider the case when $\hat{z}_i\neq0$, using \eqref{eq:Alg2_error_evol}, trigger \eqref{eq:Alg2_event2} prescribes a broadcast at time $t^*\geq t_0$ that satisfies
\begin{equation}
\textstyle (t^*-t_0)^2\hat{z}_i^2-\frac{2\sigma_i \delta_ib_ic_i}{(b_i+c_i)d_i^{out}}\hat{z}_i^2=0,
\nonumber
\end{equation}
or equivalently
\begin{equation}
\textstyle (t^*-t_0)^2=\frac{2\sigma_i \delta_ib_ic_i}{(b_i+c_i)d_i^{out}}.
\nonumber
\end{equation}
Therefore, we can lower bound the inter-events time by
\begin{equation}
\textstyle \tau_i=t^*-t_0=\sqrt{\frac{2\sigma_i \delta_ib_ic_i}{(b_i+c_i)d_i^{out}}},
\nonumber
\end{equation}
which explains our choice in \eqref{eq:Alg2_varepsilon}. By this step, if none of agent $i$'s neighbors broadcast, agent $i$ will not be triggered infinitely fast. Next, we show that messages can not be sent infinitely over a finite time period when one or more neighbors of agent $i$ trigger(s).

\textbf{Step 2}: Same as \textbf{Step 1}, assume agent $i$ has just broadcast its state at time $t_0$, thus $e_i(t_0)=0$. Our reasoning is as follows:

1) If no information is received by time $t_0 +\varepsilon_i<t_0+\tau_i$, then no trigger happens for agent $i$.

2) Let us then consider the situation that at least one neighbor of agent $i$ broadcasts its information at some time $t_1\in(t_0,t_0 +\varepsilon_i)$, which means that agent $i$ would also re-broadcast its information at time $t_1$ due to trigger \eqref{eq:Alg2_time_interval}. Define $I$ as the set in which all agents have broadcast information at time $t_1$, then as long as no agent $k \in I$ sends new information to any agent in $I$, agents in $I$ will not broadcast new information for at least $\min_{j\in I} \tau_j$ seconds, which includes the original agent $i$. As no new information is received by any agent in $I$ by time $t_1+\min_{j\in I} \varepsilon_j$, there is no problem.

3) Again consider the case that at least one agent $k$ sends new information to some agent $j \in I$ at time $t_2 \in (t_1, t_1+\min_{j\in I} \varepsilon_j)$, then by trigger \eqref{eq:Alg2_time_interval}, all agents in $I$ would also broadcast their state information at time $t_2$ and agent $k$ will now be added to $I$. The remaining reasoning is just to repeat what has been reasoned, thus, the only situation for infinite communications to occur in a finite time period is to have a network of infinite agents, which is impossible for the $N$-agent network we consider.

Therefore, \textbf{Step 1} and \textbf{Step 2} conclude that \textbf{Algorithm 2} excludes Zeno behavior for the network.
\end{proof}
\end{prop}

Next we establish the global exponential convergence.

\begin{them}
\label{them:Convergence}
(Exponential Convergence to Average Consensus). Given the system~\eqref{eq:single_dymic} executing Table~\ref{tab:algorithm} over a weight-balanced, strongly connected digraph, all agents exponentially converge to the average of their initial states, i.e. $\lim _{t\rightarrow \infty} x(t)=\bar{x}$, where $\bar{x}=\frac{1}{N}\sum_{i=1}^N x_i(0)\mathbf{1}_N$.
\end{them}

\begin{proof}
The triggering events \eqref{eq:Alg2_event1} and \eqref{eq:Alg2_event2} ensure that
\begin{equation}
\label{eq:V2_dot_bound_thrm}
\textstyle \dot{V}_2(x(t))\leq \sum_{i=1}^N (\sigma_i-1)\delta_i\Big(\sum_{j\in\mathcal{N}_i^{out}}w_{ij}(\hat{x}_i(t)-\hat{x}_j(t))\Big)^2.
\end{equation}
To show that the convergence is exponential, we show that the evolution of $V_2(x(t))$ towards $0$ is exponential. Omit the time stamp $t$ for simplicity, and define~$\sigma_{\max}=\max_{i\in\mathcal{V}}\sigma_i$, $\delta_{\max}=\max_{i\in\mathcal{V}}\delta_i$ to further bound \eqref{eq:V2_dot_bound_thrm}:
\begin{equation}
\textstyle \begin{split}
\dot{V}_2(x)&\leq(\sigma_{\max}-1)\delta_{\max} \sum_{i=1}^N \Big(\sum_{j\in\mathcal{N}_i^{out}}w_{ij}(\hat{x}_i -\hat{x}_j)\Big)^2\\
&= (\sigma_{\max}-1)\delta_{\max} \hat{x}^T L^TL\hat{x}\\
&\leq (\sigma_{\max}-1)\delta_{\max} \lambda_2(L)\hat{x}^TL^T\hat{x},
\end{split}
\nonumber
\end{equation}
where we use \eqref{eq:eigen} to come up with the last inequality.
Note that
\begin{equation}
\label{eq:Alg2_V2_bo}
\textstyle \begin{split}
V_2(x)=\frac{1}{2}x^TL^Tx &=\frac{1}{2} (\hat{x}-e)^TL^T(\hat{x}-e)\\
                     &=\frac{1}{2}(\hat{x}^TL^T\hat{x}-\hat{x}^TL^T e -e^TL^T\hat{x}+e^TL^Te)\\
                     & \leq \frac{1}{2}(2\hat{x}^TL^T\hat{x}+2e^TL^Te)\\
                     & \leq \hat{x}^TL^T\hat{x} + \|L\|\|e\|^2.
\end{split}
\end{equation}
Substitute \eqref{eq:Alg2_error_bound} into \eqref{eq:Alg2_V2_bo}, define~$d_{\min}^{out}=\min_{i\in\mathcal{V}}d_i^{out},\; b_{\max}=\max_{i\in\mathcal{V}}b_i, \; c_{\max}=\max_{i\in\mathcal{V}}c_i,\; b_{\min}=\min_{i\in\mathcal{V}}b_i$, and $c_{\min}=\min_{i\in\mathcal{V}}c_i$, using \eqref{eq:eigen}, we have
\begin{equation}
\label{eq:Alg2_V2_bound}
\textstyle \begin{split}
\hat{x}^TL^T\hat{x} + \|L\|\|e\|^2& \leq \hat{x}^TL^T\hat{x} + \|L\| \frac{2\sigma_{\max}\delta_{\max}b_{\max}c_{\max}}{(b_{\min}+c_{\min})d_{\min}^{out}}\hat{x}^TL^TL\hat{x}\\
&\leq \hat{x}^TL^T\hat{x} + \|L\| \frac{2\sigma_{\max}\delta_{\max}b_{\max}c_{\max}}{(b_{\min}+c_{\min})d_{\min}^{out}}\lambda_N(L) \hat{x}^TL^T\hat{x} \\
 &=(1+\frac{2\|L\|\sigma_{\max}\delta_{\max}b_{\max}c_{\max}\lambda_N(L)}{(b_{\min}+c_{\min})d_{\min}^{out}})\hat{x}^TL^T\hat{x}.
 \end{split}
\end{equation}       
Relate \eqref{eq:Alg2_V2_bo} with \eqref{eq:Alg2_V2_bound} gives 
\begin{equation}
\label{eq:Alg2_expConve}
\textstyle \begin{split}
\dot{V}_2(x) &\leq (\sigma_{\max}-1)\delta_{\max} \lambda_2(L)\hat{x}^TL^T\hat{x}\\
          &\leq  \frac{(\sigma_{\max}-1)\delta_{\max} \lambda_2(L)} {2(1+\frac{2\|L\|\sigma_{\max}\delta_{\max}b_{\max}c_{\max}\lambda_N(L)}{(b_{\min}+c_{\min})d_{\min}^{out}})}x^TL^Tx\\
          &=\frac{(\sigma_{\max}-1)(b_{\min}+c_{\min})\delta_{\max} \lambda_2(L)d_{\min}^{out}}{(b_{\min}+c_{\min})d_{\min}^{out}+2\|L\|\sigma_{\max}\delta_{\max}b_{\max}c_{\max}\lambda_N(L)}V_2(x).
\end{split}
\end{equation}
Substitute
$A=\frac{(\sigma_{\max}-1)(b_{\min}+c_{\min})\delta_{\max} \lambda_2(L)d_{\min}^{out}}{(b_{\min}+c_{\min})d_{\min}^{out}+2\|L\|\sigma_{\max}\delta_{\max}b_{\max}c_{\max}\lambda_N(L)}$ into \eqref{eq:Alg2_expConve}, we have $\dot{V}_2(x(t))\leq AV_2(x(t))$, therefore we conclude that $V_2(x(t))\leq V_2(x(0)) \exp(At)$ and the network converges exponentially to the average of its initial state.
\end{proof}

With the theoretical foundation of \textbf{Algorithm 2}, we are now ready to propose a class of event-triggered algorithms that can be tuned to meet varying performance needs under different scenarios.

\section{A Class of Event-Triggered Algorithms}
\label{sec:combined}
As stated in Section \ref{sec:intro}, for a given system, there are many works studying event-triggered control using Lyapunov functions to reach the goal of maintaining the stability of the system, while increasing the efficiency of the system. However, there is very little work currently available that mathematically quantifies these benefits. Recently, some works began establishing results along this line~\citep{antunes2014rollout, ramesh2016performance, khashooei2017output}, still this area is in its infancy. In particular, there are not yet established ways to compare the performance of an event-triggered algorithm with another.
Consequently, many different algorithms can be proposed to ultimately solve the same problem, while each algorithm is slightly different and produces different trajectories. Specifically in our case, \textbf{Algorithm 1} and \textbf{Algorithm 2} solve the exact same problem, and offer the exact same guarantees, i.e., they both exclude Zeno behavior and ensure asymptotic convergence of the network. So, which algorithm should we use? Moreover, we have found that depending on the initial conditions and network topology, each algorithm may out-perform the other in terms of different evaluation metrics. In any case, once these performance metrics become better researched, there will likely be more standard ways to mathematically compare the two different algorithms. Therefore, for now, instead of designing only one event-triggered algorithm for the system that only works better in one situation, we aim to design an entire class of algorithms that can easily be tuned to meet varying performance needs.

We do this by parameterizing a set of Lyapunov functions rather than studying only a specific one. To the best of our knowledge, this paper is then a first study of how to design an entire class of algorithms that use different Lyapunov functions to guarantee correctness, with the intention of being able to use the best one at all times. In this paper, we utilize only two Lyapunov functions, however, we can also use as many Lyapunov functions as we want and combine them all to develop the entire class of algorithms. 

Specifically, given any $\lambda \in [0,1]$, we define a combined Lyapunov function as 
\begin{equation}
\label{eq:Vcom}
\textstyle V_\lambda(x(t)) = \lambda V_1(x(t))+(1-\lambda)V_2(x(t)).
\end{equation}

Accordingly, the derivative of $V_\lambda(x(t))$ takes the form
\begin{equation}
\label{eq:Vcom_dot}
\textstyle \dot{V}_\lambda(x(t)) = \lambda\dot{V}_1(x(t))+(1-\lambda)\dot{V}_2(x(t)).
\end{equation}

Following the steps of deriving the triggering functions in Section \ref{sec:triggerdesign}, the triggering function developed based on the combined Lyapunov function \eqref{eq:Vcom} is given by
\begin{equation}
\label{eq:Com_Alg_Trfunc}
\textstyle \begin{split}
f_i(e_i(t))&=e_i^2(t)-\sigma_i \Big[\frac{\lambda}{4d_i^{out}}\sum_{j\in \mathcal{N}_i^{out}} w_{ij}\Big(\hat{x}_i(t)-\hat{x}_j(t)\Big)^2 +\\
&\frac{(1-\lambda) 2\delta_ib_ic_i}{(b_i+c_i)d_i^{out}}\Big(\sum_{j\in \mathcal{N}_i^{out}}w_{ij}(\hat{x}_i(t)-\hat{x}_j(t))\Big)^2\Big].
\end{split}
\end{equation}

We refer to the algorithm developed from the combined Lyapunov function as the \textbf{Combined Algorithm} parameterized by~$\lambda$, with $\lambda \in [0,1]$. Note that~$\lambda = 0$ recovers \textbf{Algorithm 2} and~$\lambda = 1$ recovers \textbf{Algorithm 1}.

Similarly, for the \textbf{Combined Algorithm}, we use the following events to avoid missing any triggers:
\begin{align}
 \quad f_i(e_i(t))&>0,\label{eq:Com_Alg_event1} \\
 f_i(e_i(t))&=0 \quad and \quad \phi_i \neq 0\label{eq:Com_Alg_event2},
\end{align}
where, with a slight abuse of notation, $\phi_i=\frac{\lambda}{4d_i^{out}}\sum_{j\in \mathcal{N}_i^{out}} w_{ij}(\hat{x}_i(t)-\hat{x}_j(t))^2 +\frac{(1-\lambda) 2\delta_ib_ic_i}{(b_i+c_i)d_i^{out}}(\sum_{j\in \mathcal{N}_i^{out}}w_{ij}(\hat{x}_i(t)-\hat{x}_j(t)))^2$. 

The parameter that bounds the inter-events time and excludes Zeno behavior is also designed:
\begin{equation}
\label{eq:Com_Alg_varepsilon}
\textstyle \varepsilon_i <\sqrt{\frac{\lambda \sigma_i}{4d_i^{out}w_i^{\max}|\mathcal{N}_i^{out}|}+\frac{2(1-\lambda)\sigma_i \delta_i b_ic_i}{(b_i+c_i)d_i^{out}}}.
\nonumber
\end{equation}

Then, with the triggering function \eqref{eq:Com_Alg_Trfunc} and $\varepsilon_i$ defined above, the \textbf{Combined Algorithm} can also be implemented using Table~\ref{tab:algorithm}.


\begin{coro}
Both \textbf{Algorithm 1} and \textbf{Algorithm 2} ensure all agents to exponentially converge to the average of their initial states with the proof that their Lyapunov functions converge exponentially. Therefore, as a linear combination of $V_1(x(t))$ and $V_2(x(t))$, $V_\lambda(x(t))$ also converges exponentially, which means that a network executing the \textbf{Combined Algorithm} shall converge exponentially to the average of its initial states.
\end{coro}

To illustrate the correctness and effectiveness of \textbf{Algorithm 2} and the \textbf{Combined Algorithm}, we introduce the fundamental clock synchronization problem that exists in wireless sensor networks (WSNs) as a case study. 

\section{Case Study: Clock Synchronization}
\label{sec:clock_sync}
\subsection{Background}
\label{subsec:background}
WSNs are broadly applied in areas such as disaster management, border protection, and security surveillance, to name a few, thanks to their low-cost and collaborative nature~\citep{abbasi2007survey, gungor2010opportunities}. However, the underlying local clocks of these sensors are often in disagreement due to the imperfections of clock oscillators. To guarantee consistency in the collected data, it is crucial to synchronize these clocks with high precision. In addition, as the small micro-processors embedded in each sensor node are usually resource-limited~\citep{gungor2010opportunities}, energy-efficient communication protocols for clock synchronization are therefore desired. 

Quite a lot approaches have been proposed to solve this problem, ranging from centralized to distributed, time-triggered to event-triggered, see~\citep{maroti2004flooding, solis2006new, simeone2007distributed, choi2010distributed, carli2014network, chen2015event, kadowaki2015event, garcia2017event} and references therein. To solve this fundamental problem, we propose to apply our event-triggered algorithms, i.e., \textbf{Algorithm 2} and the \textbf{Combined Algorithm} in this practical case. One of the most related works is done by~\citet{chen2015event}, where an event-triggered algorithm with state-dependent triggers is proposed. However, the virtual clocks they synchronize are formed in a discrete manner, which may encounter abrupt changes. The ability of avoiding abrupt changes is essential in clock synchronization since time discontinuity due to these changes can cause serious faults such as missing important events~\citep{sundararaman2005clock}. While another event-triggered algorithm proposed by~\citet{garcia2017event} does synchronize continuous-time virtual clocks, however, their time-dependent trigger design requires global information. Motivated by these two works, we introduce our state-dependent event-triggered algorithms that synchronize continuous-time virtual clocks. 

\subsection{Clock synchronization problem formulation}
\label{subsec:clk_syn}
Consider an $N$-sensor WSN whose topology is described by a strongly-connected weight-balanced underlying digraph~$\mathcal{G}=\{\mathcal{V,E},W\}$, with $\mathcal{V,\;E},\;W$ defined as in Section \ref{sec:preli}. Without loss of generality, we say that a sensor $i$ is able to receive information from its neighbors in $\mathcal{N}_i^{out}$ and send information to neighbors in $\mathcal{N}_i^{in}$. Each sensor in the network is equipped with a microprocessor with an underlying local clock $l_i(t)$, which is a function of the absolute time $t\in \mathbb{R}_{\geq 0}$. Ideally, the local clocks should be configured as $l_i(t)=t$ so that the notion of time is consistent throughout the system. In reality~\citep{kadowaki2015event}, however, they are in the form of
\begin{equation}
\label{eq:local clock}
\textstyle l_i(t)=\gamma_i t+o_i,\; i=1,\dots,N,
\end{equation}
where the unknown constants $\gamma_i\in \mathbb{R}_{>0}$ and $o_i\in \mathbb{R}$ represent the clock drift and offset of $i$-th clock, respectively.

As the absolute time $t$ is not available, the clock drift $\gamma_i$ and offset $o_i$ can not be computed directly. To synchronize the system, here we mean to synchronize the virtual clocks $T_i(t)$ of all sensors defined by~\citep{kadowaki2015event}
\begin{equation}
\label{eq:virtual_clock}
\textstyle  T_i(t)=\alpha_i(l_i(t))l_i(t),\; i=1,\dots,N,
\end{equation}
where $\alpha_i(l_i(t))$ is the controlled drift and is a function of node $i$'s local time $l_i(t)$.

The clock synchronization is said to be achieved if 
\begin{equation}
\label{eq:virtual_clock_synchr}
\textstyle \lim_{t\rightarrow \infty}|T_i(t)-T_j(t)|=0,\; \forall\; i,j\in \{1,\dots,N\}. 
\end{equation} 

For simple implementation, in this paper we consider the particular case where only clock drift is present, i.e., the clock offset $o_i=0$ for $i=1, \dots, N$. We also assume $\gamma_i\in [1-\epsilon_\gamma,1+\epsilon_\gamma]$, where $\epsilon_\gamma$ is known. The local clocks are then given by
\begin{equation}
\label{eq:local_no_drift}
\textstyle l_i(t)=\gamma_i t,\; i=1,\dots,N.
\end{equation}
Substitute~\eqref{eq:local_no_drift} into~\eqref{eq:virtual_clock} gives the expressions of virtual clocks
\begin{equation}
\label{eq:virtual_drift_sub}
\textstyle T_i(t)=\gamma_i\alpha_i(l_i(t))t,\; i=1,\dots,N.
\end{equation}

Note that the virtual clocks are continuous by definition, therefore the abrupt changes on the clocks are avoided. 

The dynamics of $\alpha_i(l_i(t))$ is specified by
\begin{equation}
\label{eq:dynamic_drift}
\textstyle \frac{{\rm d} \alpha_i(l_i(t))}{{\rm d} l_i(t)}=-\sum_{j\in\mathcal{N}_i^{out}}w_{ij}(\hat{\alpha}_i(l_i(t))-\frac{\gamma_j}{\gamma_i}\hat{\alpha}_j(l_j(t))),
\end{equation}
where $\hat{\alpha}_i(l_i(t))$, $\hat{\alpha}_j(l_j(t))$ represent the last broadcast state values of sensor $i$ and $j$ at their local time $l_i$ and $l_j$, respectively. Though $\gamma_i$ and $\gamma_j$ can not be computed directly, the value of $\frac{\gamma_j}{\gamma_i}$ can be obtained as follows~\citep{garcia2017event}: record the local time of node $i$ and node $j$ when node $i$ receives information from node $j$ at two time points, say $t_m$ and $t_n$, then $\frac{a_j}{a_i}$ can be computed using $\frac{\gamma_j}{\gamma_i}=\frac{l_j(t_m)-l_j(t_n)}{l_i(t_m)-l_i(t_n)}$. Note we only need the local clock time, not the exact values of $t_m$ and $t_n$.

Define $e_i(l_i(t))=\hat{\alpha}_i(l_i(t))-\alpha_i(l_i(t))$ as sensor $i$'s state error, where $\alpha_i(l_i(t))$ is its current controlled drift. An event for sensor $i$ is triggered as soon as the triggering function
\begin{equation}
\label{eq:clk_trig_func}
\textstyle f_i(l_i(t),\alpha_i(l_i(t)),\hat{\alpha}_i(l_i(t)),\hat{\alpha}_j(l_j(t)) )>0
\end{equation}
is fulfilled. The triggered event then drives sensor $i$ to broadcast its current state $\alpha_i(l_i(t))$ to its neighbors so that they can update their states accordingly. Our objective is to apply \textbf{Algorithm 2} and the \textbf{Combined Algorithm} so as to design triggering functions \eqref{eq:clk_trig_func} for each sensor with its locally available information so that the virtual clocks are synchronized, i.e., \eqref{eq:virtual_clock_synchr} is satisfied.
 
\subsection{Distributed event-triggered clock synchronization algorithms}
\label{subsec:dis_clk_syn_alg}
The event-triggered algorithms for clock synchronization are developed based on Lyapunov functions. To begin, let us first rewrite \eqref{eq:virtual_drift_sub} as
\begin{equation}
\textstyle  T_i(t)=\gamma_i\alpha_i(l_i(t))t=y_i(t)t, 
\end{equation}
where $y_i(t)=\gamma_i\alpha_i(l_i(t))$ is called the modified drift. It is clear that once consensus is achieved on the variables $y_i(t)$, the clock synchronization will be realized regardless of the individual values of $\gamma_i$ and $\alpha_i(l_i(t))$. 

We then adopt the Lyapunov candidate functions proposed in Section \ref{sec:triggerdesign}, with the modified drifts as variables, i.e., $V_1(y(t))=\frac{1}{2}(y(t)-\bar{y})^T(y(t)-\bar{y})$, $V_2(y(t))=\frac{1}{2}y(t)^TL^Ty(t)$, and $V_\lambda(y(t)) = \lambda V_1(y(t))+(1-\lambda)V_2(y(t))$. As the algorithm development with different Lyapunov functions are similar, we only use $V_2(y(t))=\frac{1}{2}y(t)^TL^Ty(t)$ as an example to illustrate the derivation process. 

The dynamics of the modified drift $y_i(t)$ is derived as follows:
\begin{equation}
\label{eq:dot_x}
\textstyle \begin{split}
\dot{y_i}(t)=\frac{{\rm d} \alpha_i(l_i(t))}{{\rm d} l_i(t)}\cdot\frac{{\rm d}l_i(t)}{{\rm d} t}&=\gamma_i^2\Big(-\sum_{j\in\mathcal{N}_i^{out}}w_{ij}(\hat{\alpha}_i(l_i(t))-\frac{\gamma_j}{\gamma_i}\hat{\alpha}_j(l_j(t)))\Big)\\
&=-\gamma_i\Big(\sum_{j\in\mathcal{N}_i^{out}}w_{ij}(\gamma_i\hat{\alpha}_i(l_i(t))-\gamma_j\hat{\alpha}_j(l_j(t))\Big)\\
&=-\gamma_i\sum_{j\in\mathcal{N}_i^{out}}w_{ij}\Big(\hat{y}_i(t)-\hat{y}_j(t)\Big).
\end{split}
\end{equation} 

We then specify the following Lemma to upper bound the derivatives of $V_2(y(t))$. 


\begin{lemma}
\label{lemma:clk_V2}
In clock synchronization, for $i \in\mathcal{V}$, let $b_i,c_j>0$ for all $i,j \in \mathcal{V}$ (the same $b_i, c_j$ as in Lemma \ref{lemma:V2_dot_upbound},  define $\nu_i(t)= \sum_{j\in\mathcal{N}_i^{out}}w_{ij}(\hat{y}_i(t)-\hat{y}_j(t))$, and $e_{yi}(t)=\hat{y}_i(t)-y_i(t)$, then the derivative of $V_2(y(t))=\frac{1}{2}y(t)^TL^Ty(t)$ is upper bounded by
\begin{equation}
\label{eq:clk_V2_dot_bound}
\textstyle \dot{V}_2(y(t))\leq -\sum_{i=1}^N \gamma_i\left[\delta_i \Big(\sum_{j\in\mathcal{N}_i^{out}}w_{ij}(\hat{y}_i(t)-\hat{y}_j(t))\Big)^2-\left(\frac{d_i^{out}}{2b_i}+\frac{d_i^{out}}{2c_i}\right)e_{yi}^2(t) \right],
\end{equation}
where $\delta_i$ is what defined in \eqref{eq:delta}.
\end{lemma}
The proof is similar to the proof for Lemma \ref{lemma:V2_dot_upbound} and is omitted due to space limit. 

From Lemma \ref{lemma:clk_V2}, we can see that as long as $\dot{V}_2(y(t))<0$ and $\|Ly(t)\|\neq 0$ hold, $y_i(t)$ achieves consensus, meaning $\lim_{t\rightarrow \infty}|y_i(t)-y_j(t)|=0$. Recall that $T_i(t)=y_i(t)t$, therefore, $\lim_{t\rightarrow \infty}|T_i(t)-T_j(t)|=0$, proving that the synchronization on virtual clocks can be achieved. 


A sufficient condition to ensure that $V_2(y(t))$ is monotonically decreasing is 
\begin{equation}
\label{eq:clk_Alg2_errbound}
\begin{split}
e_{yi}^2(t) &\leq \frac{\delta_i}{(\frac{d_i^{out}}{2b_i}+\frac{d_i^{out}}{2c_i})}\Big(\sum_{j\in\mathcal{N}_i^{out}}w_{ij}(\hat{y}_i(t)-\hat{y}_j(t))\Big)^2\\
&=\frac{2b_ic_i\delta_i}{(b_i+c_i)d_i^{out}}\Big(\sum_{j\in\mathcal{N}_i^{out}}w_{ij}(\gamma_i\hat{\alpha}_i(l_i(t))-\gamma_j\hat{\alpha}_j(l_i(t)))\Big)^2\\
&=\frac{2\gamma_i^2b_ic_i\delta_i}{(b_i+c_i)d_i^{out}}\Big(\sum_{j\in\mathcal{N}_i^{out}}w_{ij}(\hat{\alpha}_i(l_i(t))-\frac{\gamma_j}{\gamma_i}\hat{\alpha}_j(l_j(t)))\Big)^2.\\
\end{split}
\end{equation}

With $e_{yi}(t)=\gamma_ie_i(l_i(t))$, we define the triggering function developed from \textbf{Algorithm 2} as
\begin{equation}
\label{eq:clk_Alg2_triggFunc}
f_i(e_i(l_i(t)))=e_i^2(l_i(t))-\frac{2b_ic_i\delta_i}{(b_i+c_i)d_i^{out}}\Big(\sum_{j\in\mathcal{N}_i^{out}}w_{ij}(\hat{\alpha}_i(l_i(t))-\frac{\gamma_j}{\gamma_i}\hat{\alpha}_j(l_j(t)))\Big)^2.
\end{equation}

To ensure no triggers are missed by sensor $i$, we define an event either by 
\begin{align}
f_i(e_i(l_i(t)))&>0 \quad or\label{eq:clk_Alg2_event1}\\
f_i(e_i(l_i(t)))=0 \quad and \quad &\sum_{j\in\mathcal{N}_i^{out}}w_{ij}(\hat{\alpha}_i(l_i(t))-\frac{\gamma_j}{\gamma_i}\hat{\alpha}_j(l_j(t))) \neq 0\label{eq:clk_Alg2_event2}.
\end{align}

Similarly, an additional trigger is prescribed to address the non-Zeno behavior. Let $l_i^{last}$ be the last time at which sensor $i$ broadcasts its information to its neighbors. If at some time $l_i(t)\geq l_i^{last}$, sensor $i$ receives information from a neighbor $j\in \mathcal{N}_i^{out}$, then it immediately broadcasts its state if
\begin{equation}
\label{eq:clk_Alg2_inter-event}
l_i(t)\in (l_i^{last}, l_i^{last}+\varepsilon_i'),
\end{equation}
where
\begin{equation}
\label{eq:clk_Alg2_varep}
\varepsilon_i' < \sqrt{\frac{2\sigma_ib_ic_i\delta_i}{(b_i+c_i)d_i^{out}}}
\end{equation}
whose design is as given in Proposition \ref{prop:non-Zeno}.

The following result presents \textbf{Algorithm 2} in the clock synchronization application. 
\begin{them}
\label{them:clk_Alg2_sync}
For an $N$-sensor network over a weight-balanced digraph, assume only clock drift exists, i.e.,\;$o_i = 0,\; \forall i\in \mathcal{V}$. With the virtual clocks~\eqref{eq:virtual_clock}, dynamics given in~\eqref{eq:dynamic_drift}, the distributed event-triggered consensus algorithm~\eqref{eq:clk_Alg2_triggFunc}-\eqref{eq:clk_Alg2_varep} (\textbf{Algorithm 2}) achieves asymptotic synchronization for the virtual clocks, i.e.,~\eqref{eq:virtual_clock_synchr} is satisfied.
\end{them}

We haven shown that \textbf{Algorithm 2} can be applied to the practical clock synchronization problem. Next, we show that the \textbf{Combined Algorithm} can also be applied to solve the clock synchronization problem. To do so, we first derive the triggering law for the clock synchronization problem from \textbf{Algorithm 1} as
\begin{equation}
\label{eq:clk_Alg1_triggFunc}
f_i(e_i(l_i(t)))= e_i^2(l_i(t))-\frac{\sigma_i}{4d_i^{out}}\sum_{j\in \mathcal{N}_i^{out}}w_{ij}(\hat{\alpha}_i(l_i(t))-\frac{\gamma_j}{\gamma_i}\hat{\alpha}_j(l_j(t)))^2,
\end{equation}
with an inter-event period bounded by $\varepsilon_i' < \sqrt{\frac{\sigma_i}{4d_i^{out}w_i^{\max}|\mathcal{N}_i^{out}|}}$.

Then, with the triggering rules \eqref{eq:clk_Alg2_triggFunc} - \eqref{eq:clk_Alg1_triggFunc} and the analysis in Section \ref{sec:combined}, designing the triggering function for the clock synchronization problem from the \textbf{Combined Algorithm} is straightforward. That is, 
\begin{equation}
\label{eq:clk_ComAlg_triggFunc}
\begin{split}
f_i(e_i(l_i(t)))= e_i^2(l_i(t))&-\frac{\lambda\sigma_i}{4d_i^{out}}\sum_{j\in \mathcal{N}_i^{out}}w_{ij}(\hat{\alpha}_i(l_i(t))-\frac{\gamma_j}{\gamma_i}\hat{\alpha}_j(l_j(t)))^2\\
&-\frac{2(1-\lambda)\sigma_i\delta_ib_ic_i}{(b_i+c_i)d_i^{out}}\Big(\sum_{j\in\mathcal{N}_i^{out}}w_{ij}(\hat{\alpha}_i(l_i)-\frac{\gamma_j}{\gamma_i}\hat{\alpha}_j(l_j))\Big)^2 ,
\end{split}
\end{equation}
with an inter-event period bounded by $\varepsilon_i' < \sqrt{\frac{\lambda\sigma_i}{4d_i^{out}w_i^{\max}|\mathcal{N}_i^{out}|}+\frac{2(1-\lambda)\sigma_i \delta_ib_ic_i}{(b_i+c_i)d_i^{out}}}$.

\begin{them}
\label{them:clk_Comb_sync}
For an $N$-sensor network over a weight-balanced digraph, assume only clock drift exists, i.e.,\;$o_i = 0,\; \forall i\in \mathcal{V}$. With the virtual clocks~\eqref{eq:virtual_clock}, dynamics given in~\eqref{eq:dynamic_drift}, the distributed event trigging rule defined in~\eqref{eq:clk_ComAlg_triggFunc}, then the \textbf{Combined Algorithm} achieves asymptotic synchronization for the virtual clocks when the triggering condition $f_i(e_i(l_i(t)))>0$ or $f_i(e_i(l_i(t)))=0$ with $ e_i^2(l_i(t))>0$ is met.
\end{them}
The proof of the theorem and the stability analysis, non-Zeno behavior exclusion are as given in Section \ref{sec:analysis}, therefore are omitted.

\section{Simulation Results}
\label{sec:sim}
In this section, we apply \textbf{Algorithm 1} and \textbf{Algorithm 2} to the event-triggered clock synchronization problem, to show the effectiveness of both algorithms. We then demonstrate the performance of the proposed algorithms through several simulations and show how either \textbf{Algorithm 1} or \textbf{Algorithm 2} could be argued to be `better' given different network topology, which has set the basis for our introduction of the \textbf{Combined Algorithm} to easily go between the two.

We first show that both \textbf{Algorithm 1} and \textbf{Algorithm 2} are able to synchronize the virtual clocks in WSNs. We consider four different network topologies, with their corresponding weighted adjacency matrices listed in Table~\ref{tab:topology}.  
\begin{table}
	\centering
	\caption{Four different networks.}
{\begin{tabular}{c|c}	\hline 

\textbf{Network 1: Random network}	& \textbf{Network 2: Ring network}    \\ 
$W_1=  \begin{bmatrix} 
0   &  1  &  0   & 0  &  0 \\
0 &  0 &  1/2  & 1/2  & 0\\
5/6 &  0 &  1/6 &  0  & 0 \\
1/6  & 0 &  1/6  & 1/2  & 1/6\\
0  & 0  & 1/6  & 0 &  5/6
\end{bmatrix}$ &
	
$W_2=  \begin{bmatrix} 
0  &  1/2  &  0   & 0  &  1/2 \\
1/2   &   0   &  1/2  &  0  &  0\\
0  &  1/2  &  0  &   1/2   &   0 \\
0  &  0  &  1/2   &  0  &  1/2\\
1/2   &  0  &  0  &  1/2  &  0
\end{bmatrix}$ \\ \hline 
	
\textbf{Network 3: Complete network} & \textbf{Network 4: Star network} \\

$W_3=  \begin{bmatrix} 
0   &  1/4  &  1/4    & 1/4   &  1/4  \\
1/4    &   0   &  1/4  & 1/4  &  1/4 \\
1/4   &  1/4   &  0  &   1/4    &   1/4  \\
1/4  &  1/4   & 1/4    &  0 &  1/4\\
1/4    &  1/4  &1/4   &  1/4   &  0
\end{bmatrix}$

& $W_4=  \begin{bmatrix} 
0   &  1/4  &  1/4    & 1/4   &  1/4 \\
1/4   &   0   & 0  &  0  &  0\\
1/4  & 0  &  0 &   0   &   0 \\
1/4  & 0  &   0   & 0  &  0\\
1/4   &  0  &  0 &  0  &  0
\end{bmatrix}$ \\
	\hline	
	\end{tabular}}
	\label{tab:topology}
\end{table}

The clock offset is $0$ for all nodes, and the unknown clock drifts are $\gamma=[0.65\; 0.79\; 0.91\; 1.25\; 1.4]^T$. The evolution of the local clocks with respect to the absolute time $t$ is shown in Figure~\ref{fig:local_no_offset}. We can see that without any control, the local clocks will diverge.  

\begin{figure}
	\centering
	\subfloat[Local clocks]{
		\label{fig:local_no_offset}
		\resizebox*{6.5cm}{!}{\includegraphics{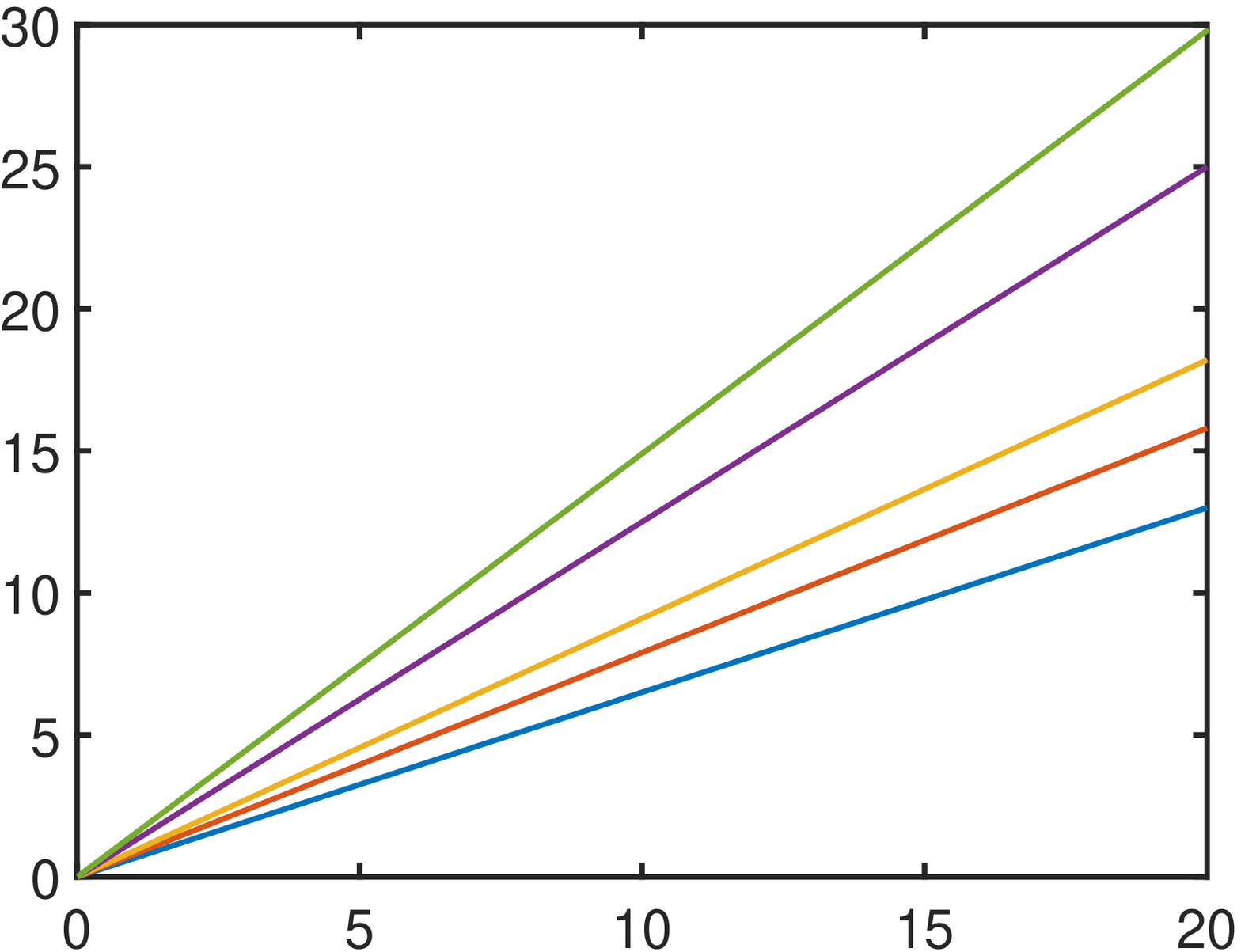}}
		\put(-196,80){\scriptsize $l(t)$}
		\put(-90,-7){\scriptsize $t$}}\hspace{20pt} 
	\subfloat[Virtual clocks]{
		\label{fig:virtual_no_offset}
		\resizebox*{6.5cm}{!}{\includegraphics{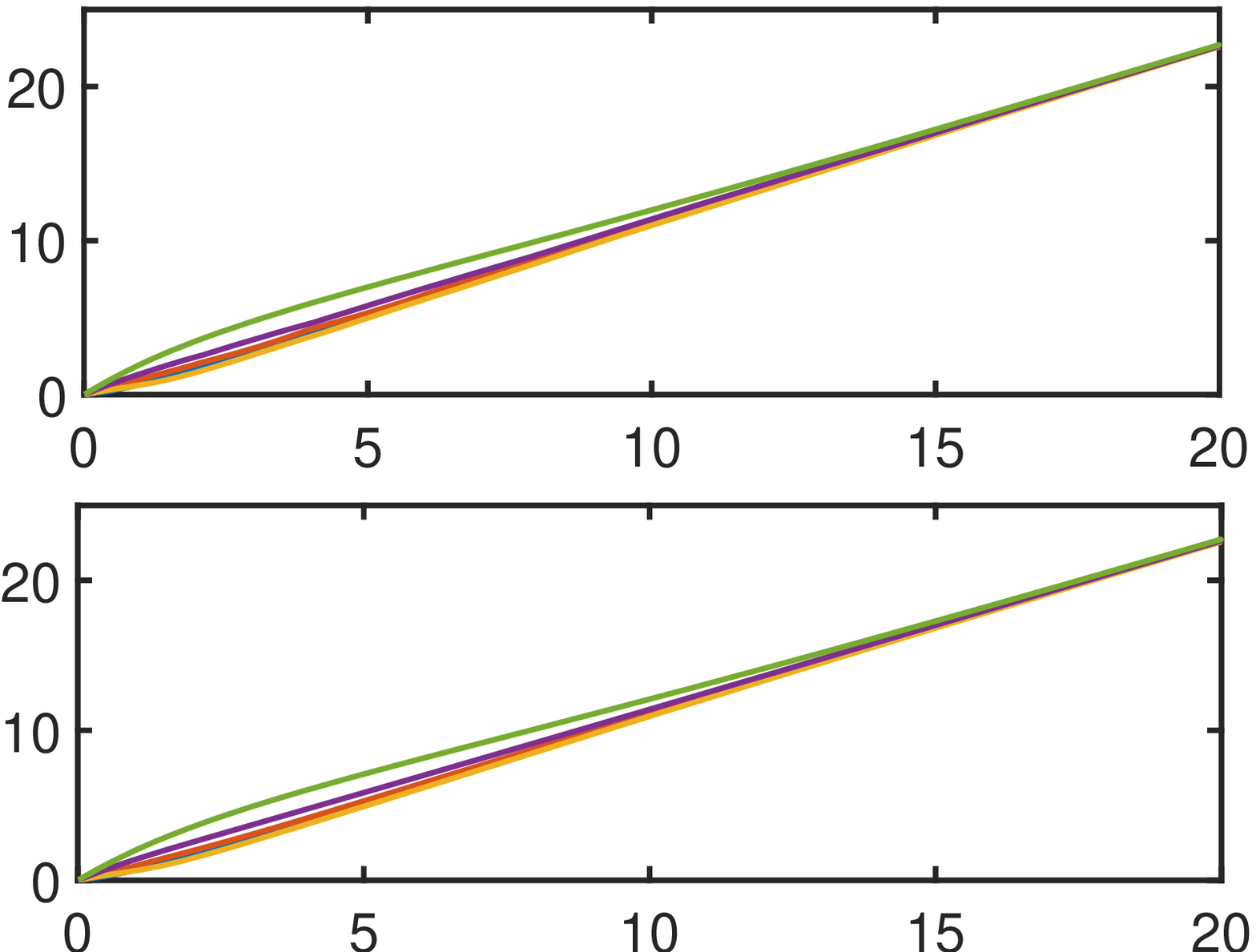}}
		\put(-200,102){\scriptsize $T(t)$}
		\put(-200,35){\scriptsize $T(t)$}
		\put(-90,-7){\scriptsize $t$}}
 
	\caption{Plots of the simulation results of the clock synchronization on Network 1. (a) The local clocks are the same for both algorithms. (b) Virtual clocks with the implementation of event-triggered control. Both \textbf{Algorithm 1} (top) and \textbf{Algorithm 2} (bottom) are able to synchronize the virtual clocks.}
	\label{fig:clocks}
\end{figure}

Then, we implement \textbf{Algorithm 1} and \textbf{Algorithm 2} with the control law \eqref{eq:dynamic_drift}, triggering functions \eqref{eq:clk_Alg1_triggFunc} and \eqref{eq:clk_Alg2_triggFunc} developed from \textbf{Algorithm 1} and \textbf{Algorithm 2}, respectively, to achieve clock synchronization. The involved parameters are set to be $\sigma_i = 0.5$, and $b_i=c_i=0.5/d_i^{out}$ for all $i\in\{1,\dots,N\}$. Both algorithms are able to synchronize the virtual clocks on all four networks and we take the result on \textbf{Network 1} as an example and show the virtual clock evolution in Figure~\ref{fig:virtual_no_offset}. However, except for the synchronization, we have no idea which algorithm performs better on other evaluation metrics, for example, the convergence speed and total energy consumption. Also, the performance evaluation result may differ for different network topologies. Therefore, in the following simulations, we show the difference of the two different algorithms on four network topologies with different evaluation metrics.

We plot the triggering instances of all nodes in the network when implementing the event-triggered algorithms in Figure \ref{fig:trigg_inst_rand}, \ref{fig:trigg_inst_ring}, \ref{fig:trigg_inst_comp}, \ref{fig:trigg_inst_star}. We notice that in general, the number of events triggered when implementing \textbf{Algorithm 1} is less than that when implementing \textbf{Algorithm 2}. We also plot the evolution of Lyapunov functions, i.e., $V_1(y(t))=\frac{1}{2}(y(t)-\bar{y})^T(y(t)-\bar{y})$, $V_2(y(t))=\frac{1}{2}y(t)^TL^Ty(t)$, and $V_\lambda(y(t)) = \lambda V_1(y(t))+(1-\lambda)V_2(y(t))$ for all networks with $\sigma_i=0.5$ for all agents in Figure \ref{fig:lyapun_rand}, \ref{fig:lyapun_ring}, \ref{fig:lyapun_comp}, \ref{fig:lyapun_star}, which again corroborates our analysis that both algorithms ensure convergence, or in this case, synchronization for the resulting systems. We can see that except for \textbf{Network 3}, the Lyapunov function of \textbf{Algorithm 2} in the other three networks converges faster than that of \textbf{Algorithm 1}. It is also noted that when the number of events triggered when implementing \textbf{Algorithm 2} is noticeably larger than that in implementing \textbf{Algorithm 1}, the convergence speed of the Lyapunov function in \textbf{Algorithm 2} is also noticeably faster than that in \textbf{Algorithm 1}. This is reasonable, since the more events are triggered, the more information is communicated in the network, and the faster the consensus will be reached.

\begin{figure}
	\centering
	\subfloat[Network 1, triggering instances]{
		\label{fig:trigg_inst_rand}
		\resizebox*{6.5cm}{!}{\includegraphics{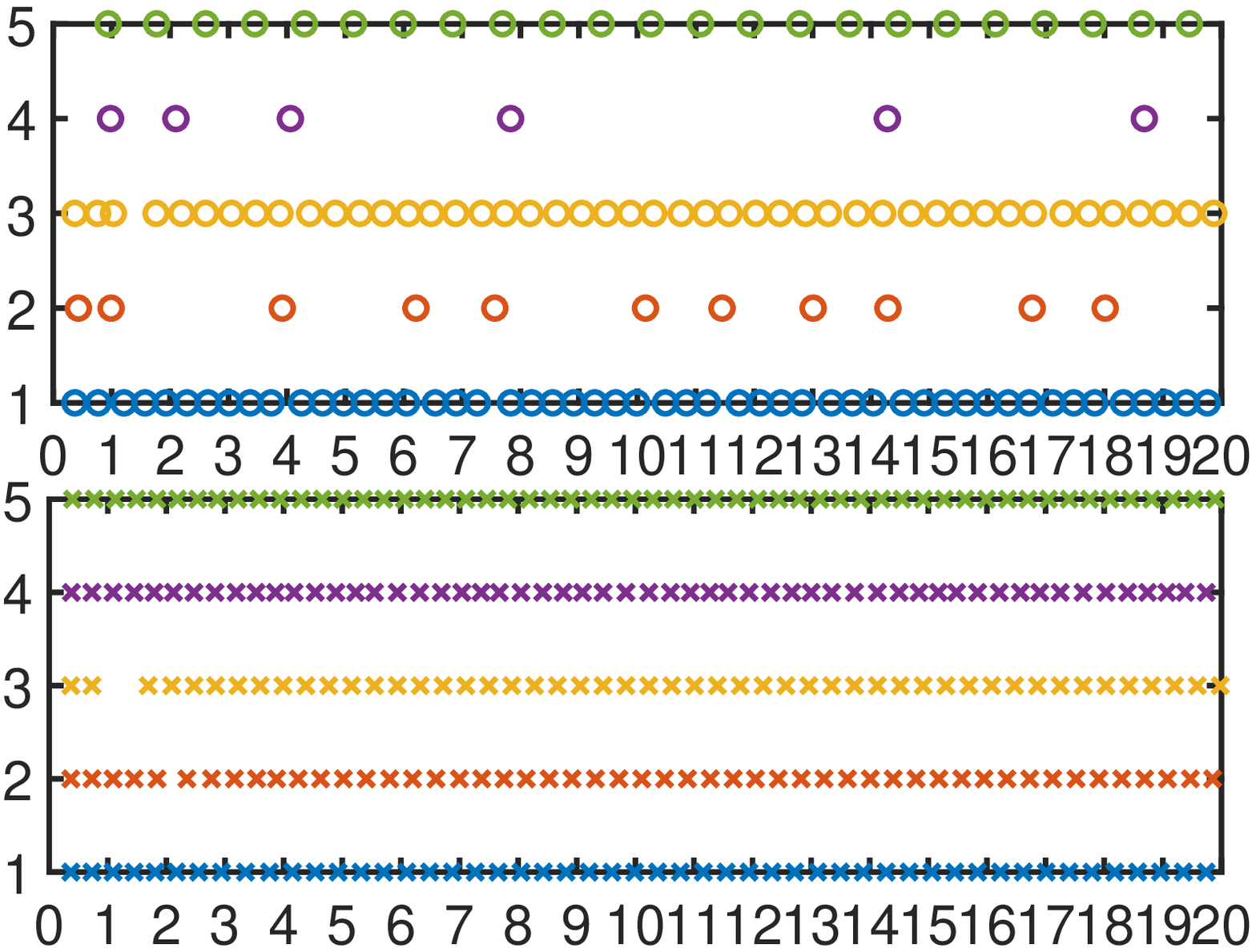}}
		\put(-210,105){\scriptsize Events}
		\put(-210,35){\scriptsize Events}
		\put(-90,-7){\scriptsize $t$}}\hspace{20pt} 
	\subfloat[Network 1, Lyapunov function]{
		\label{fig:lyapun_rand}
		\resizebox*{6.5cm}{!}{\includegraphics{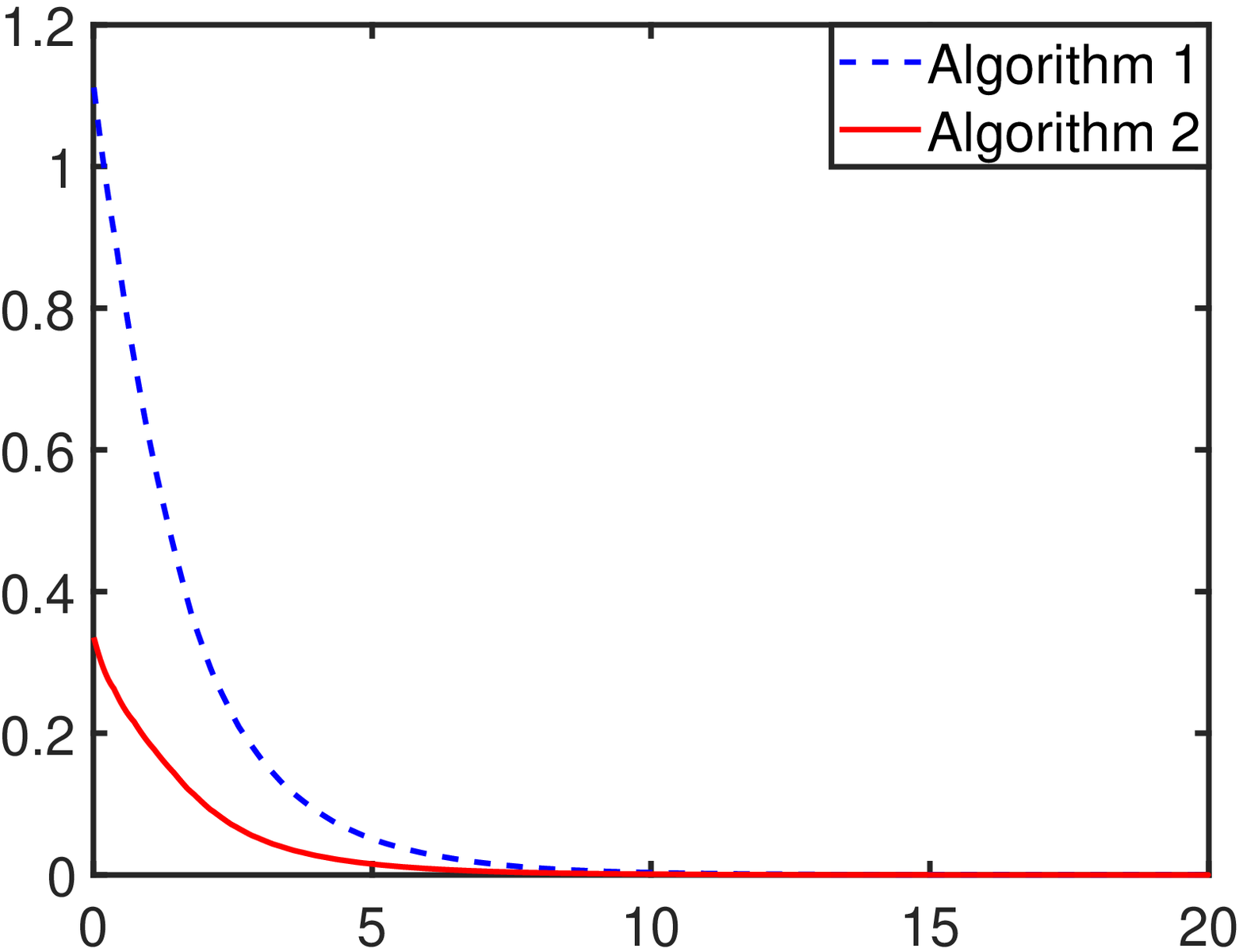}}
		\put(-202,70){\scriptsize $V(t)$}
		\put(-90,-7){\scriptsize $t$}}\\
	\subfloat[Network 2, triggering instances]{
		\label{fig:trigg_inst_ring}
		\resizebox*{6.5cm}{!}{\includegraphics{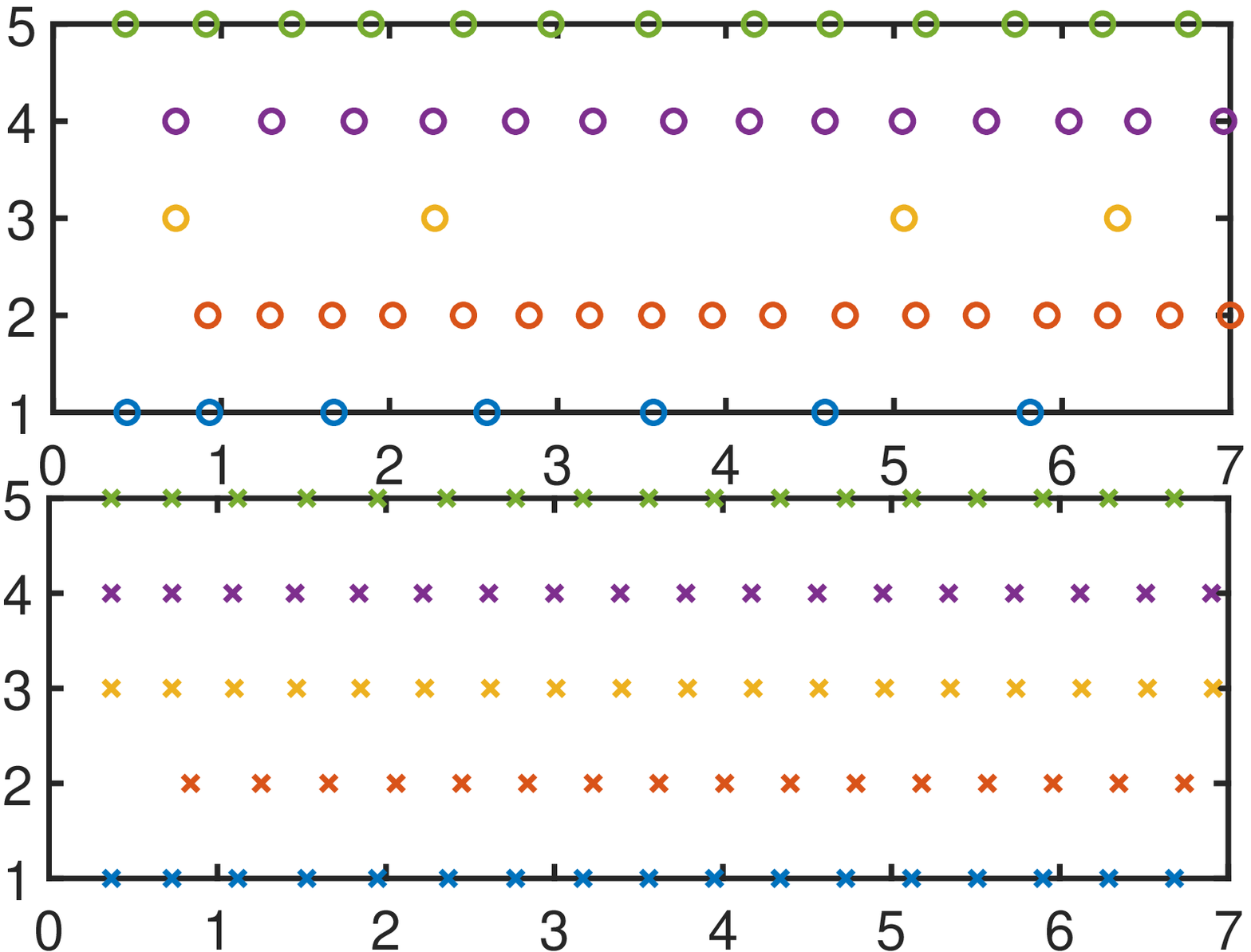}}
		\put(-210,105){\scriptsize Events}
		\put(-210,35){\scriptsize Events}
		\put(-90,-7){\scriptsize $t$}}\hspace{20pt} 
	\subfloat[Network 2, Lyapunov function]{
		\label{fig:lyapun_ring}
		\resizebox*{6.5cm}{!}{\includegraphics{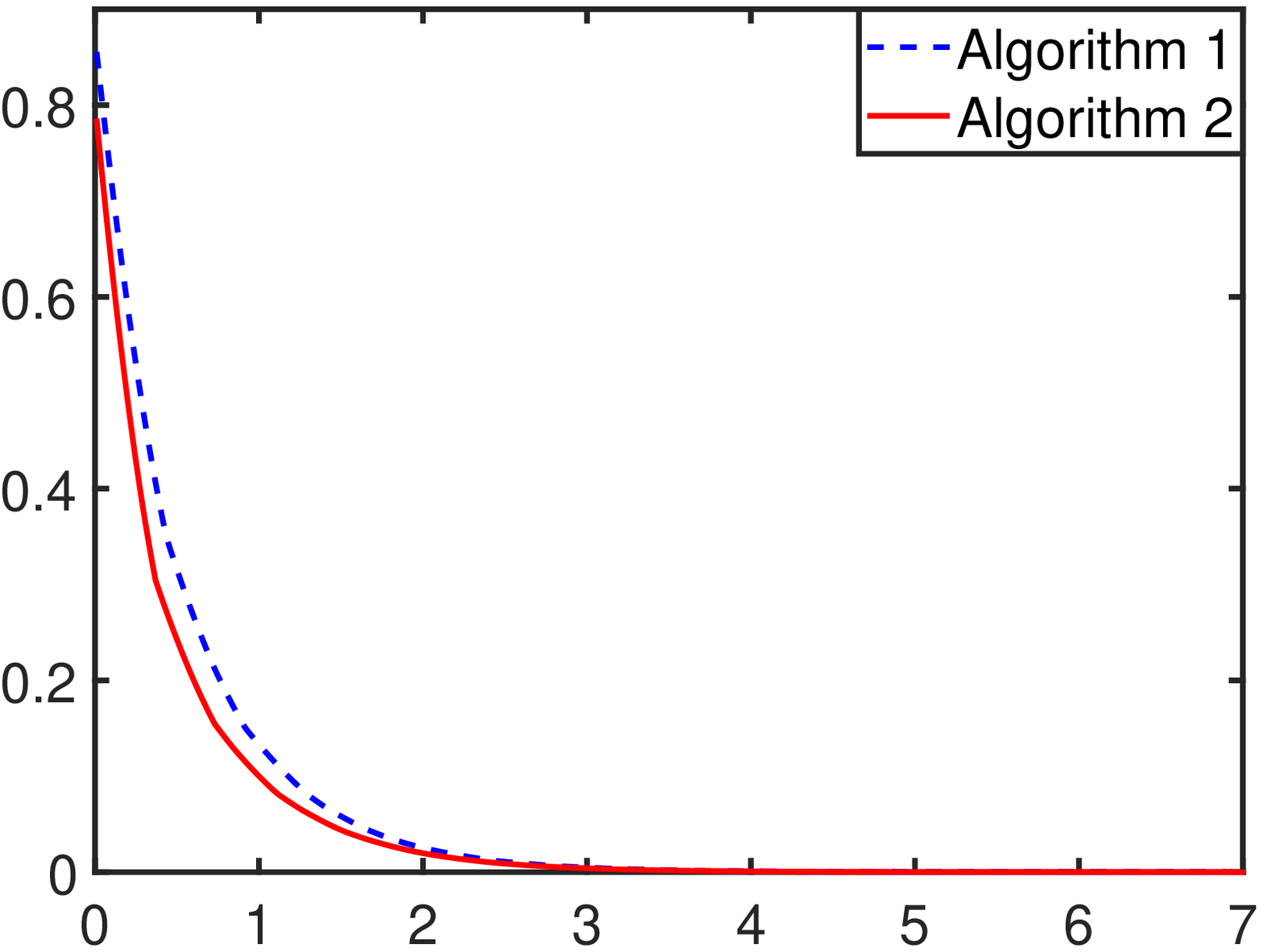}}
		\put(-202,70){\scriptsize $V(t)$}
		\put(-90,-7){\scriptsize $t$}}\\
	\subfloat[Network 3, triggering instances]{
		\label{fig:trigg_inst_comp}
		\resizebox*{6.5cm}{!}{\includegraphics{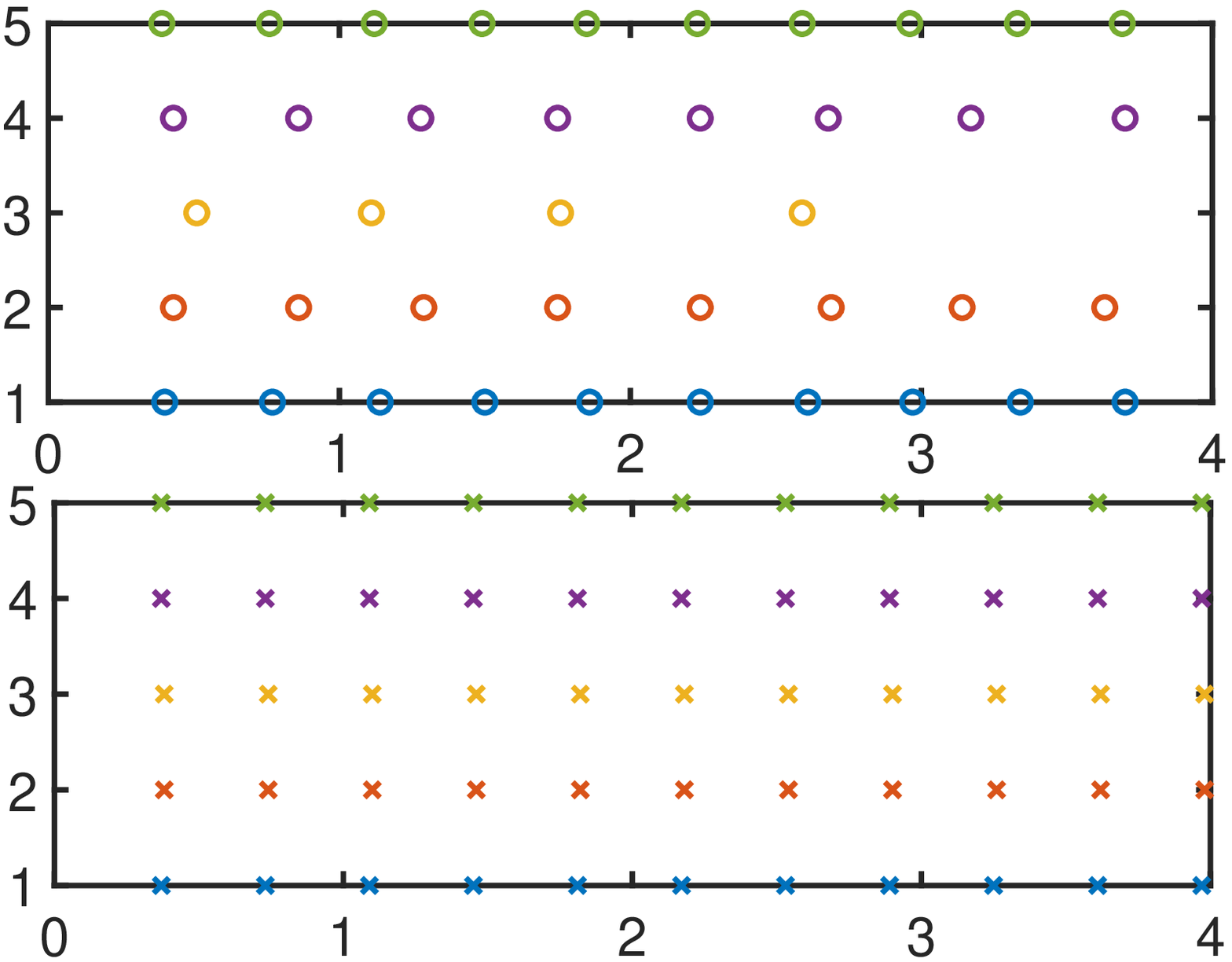}}
	    \put(-210,105){\scriptsize Events}
	    \put(-210,35){\scriptsize Events}
		\put(-90,-7){\scriptsize $t$}}\hspace{20pt} 
	\subfloat[Network 3, Lyapunov function]{
		\label{fig:lyapun_comp}
		\resizebox*{6.5cm}{!}{\includegraphics{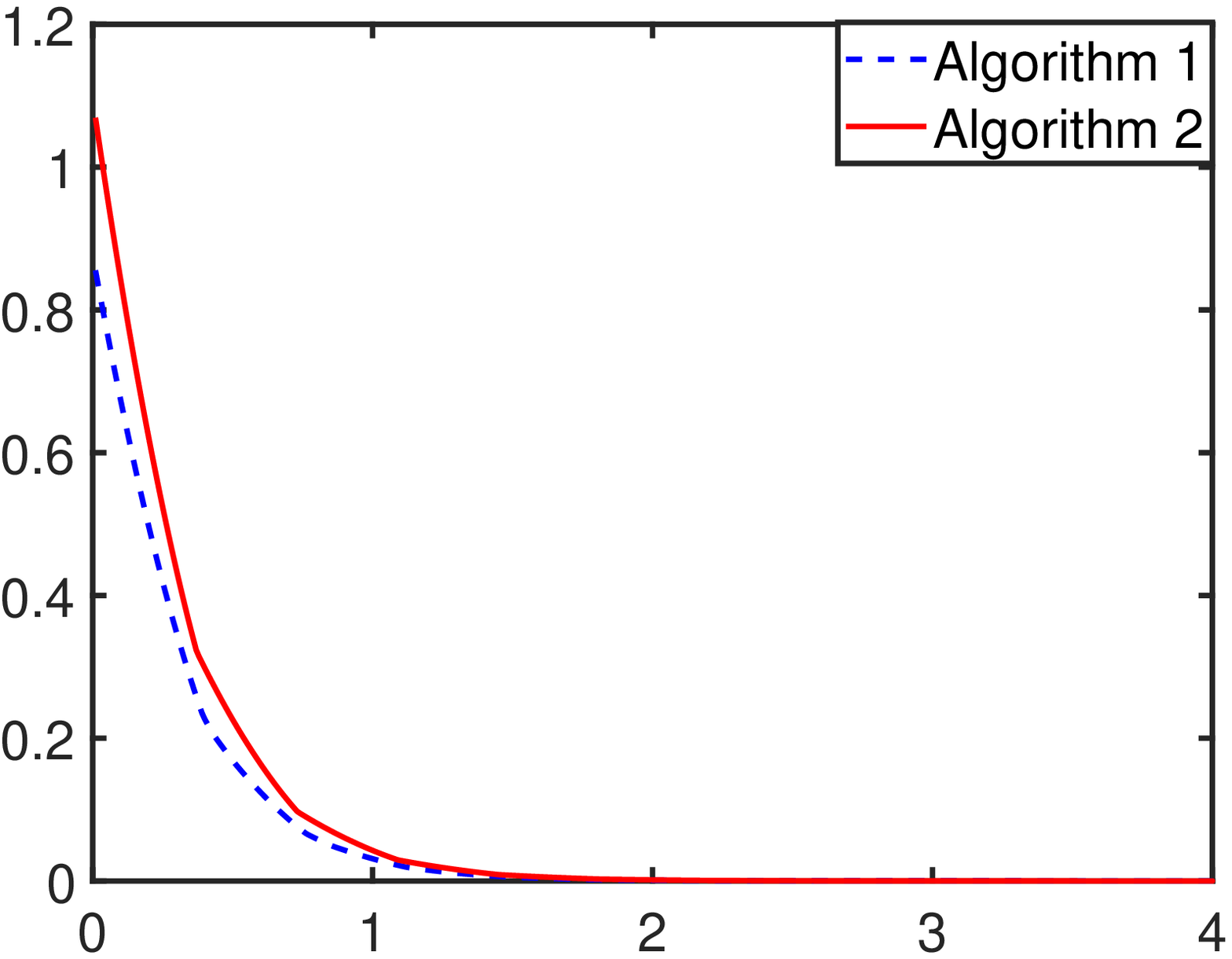}}
		\put(-202,70){\scriptsize $V(t)$}
		\put(-90,-7){\scriptsize $t$}}\\
	\subfloat[Network 4, triggering instances]{
		\label{fig:trigg_inst_star}
		\resizebox*{6.5cm}{!}{\includegraphics{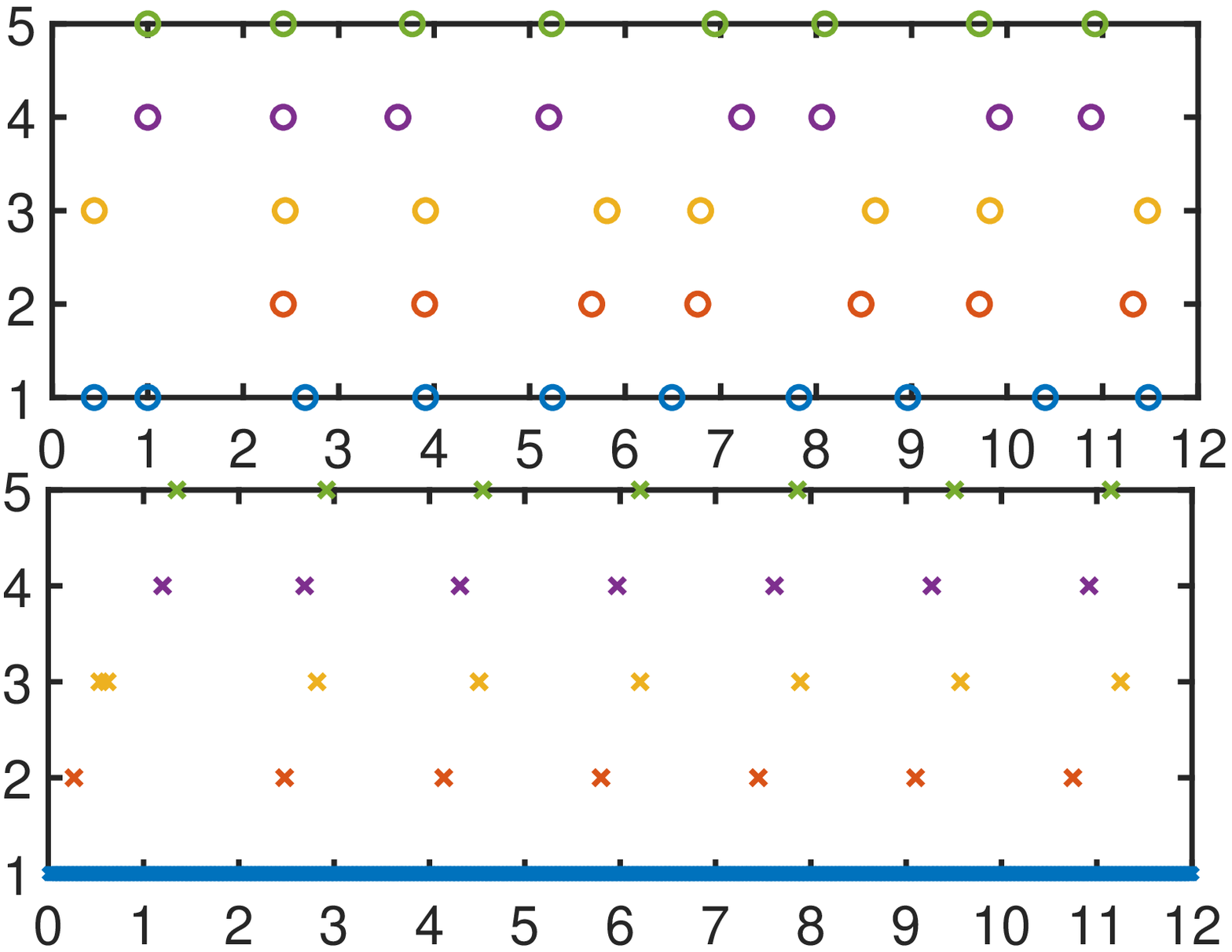}}
		\put(-210,105){\scriptsize Events}
		\put(-210,35){\scriptsize Events}
		\put(-90,-7){\scriptsize $t$}}\hspace{20pt} 
	\subfloat[Network 4, Lyapunov function]{
		\label{fig:lyapun_star}
		\resizebox*{6.5cm}{!}{\includegraphics{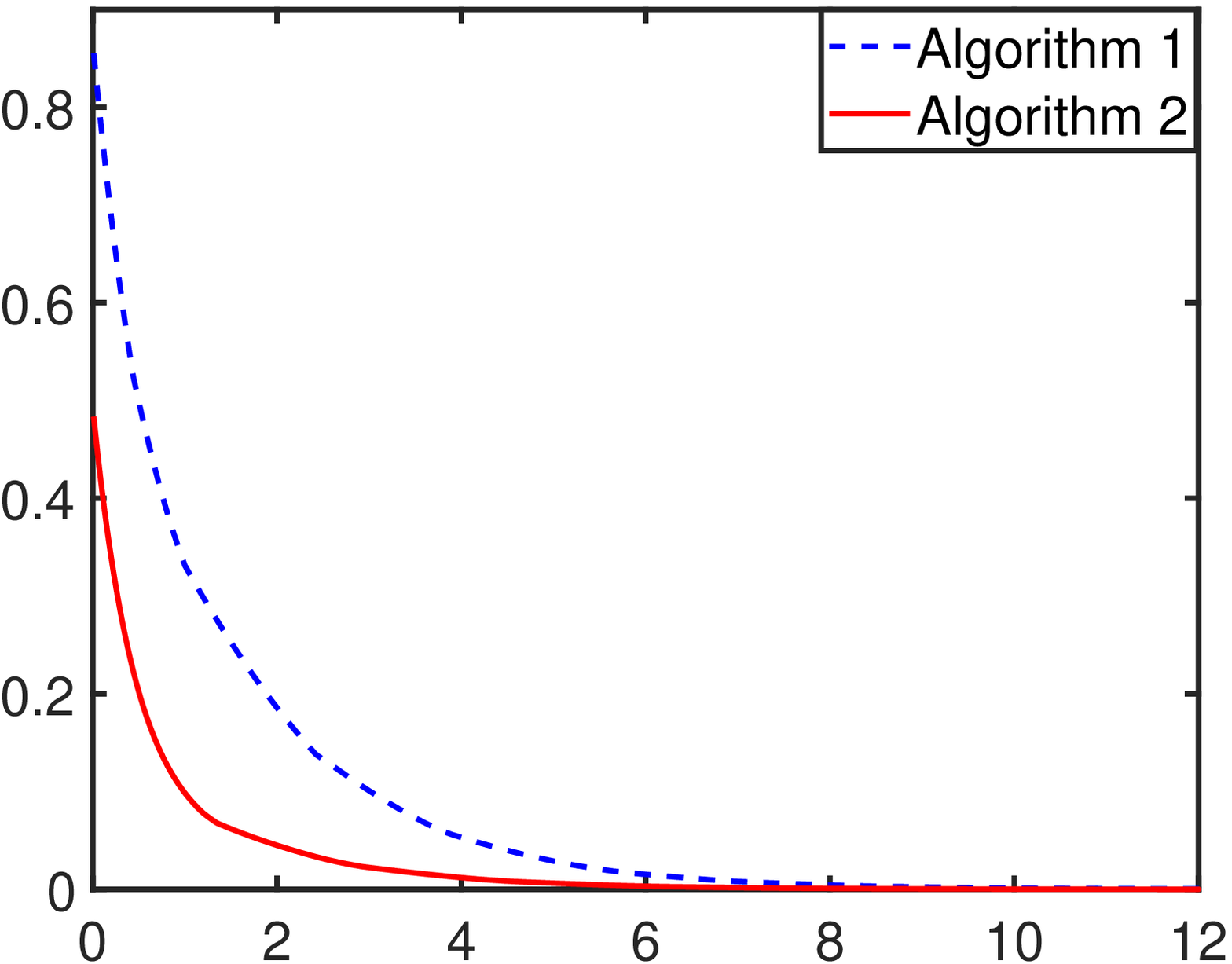}}
		\put(-202,70){\scriptsize $V(t)$}
		\put(-90,-7){\scriptsize $t$}}
	\caption{Plots of the triggering instances and the evolution of Lyapunov candidate functions on four networks when implementing both Algorithms. For figure (a), (c), (e), (g), \textbf{Algorithm 1} is on the top and \textbf{Algorithm 2} is on  the bottom.}
	\label{fig:trigger_lyap}
\end{figure}

The above simulations corroborate our argument that depending on the chosen evaluation metric, either algorithm can be argued to be `better' than the other. To better quantize/visualize the performance difference of the two algorithms and demonstrate our motivations for proposing the \textbf{Combined Algorithm}, we then executing all algorithms with respect to varying~$\sigma$. We evaluate these algorithms with four performance metrics, 1) the total number of events triggered, denoted by $N_e$, 2) the time needed for each network to reach a 99\% convergence of the Lyapunov function, denoted by  $T_{con}$, 3) the total communication energy required to achieve a 99\% convergence, denoted by $E$, and 4) the square of the $H_2$-norm of the system, denoted by $\mathcal{C}$~\citep{dezfulian2018performance}. The total communication energy needed is calculated by multiplying the power in units of milliwatt (mW) with $T_{con}$, where we adopt the following power calculation model in units of mW~\citep{martins2008jointly}:
\begin{equation}
\mathcal{P}= \sum_{i=1}^N\left[\sum_{j\in\{1,\dots,N\},j\neq i}\eta 10^{0.1P_{i\rightarrow j}+\zeta \|\alpha_i(l_i(t))-\alpha_j(l_j(t))\|}\right],
\nonumber
\end{equation}
where~$\zeta>0$ and~$\eta>0$ depend on the characteristics of the wireless medium and $P_{i\rightarrow j}$ is the power of the signal transmitted from agent $i$ to agent $j$ in units of dBmW. Similar as \citep{nowzari2012self}, we set $\eta$, $\zeta$ and $P_{i\rightarrow j}$ to be $1$. The square of the $H_2$-norm, $\mathcal{C}$ is defined by
\begin{equation}
 \mathcal{C}:=\int_{t=0}^{\infty} \sum_{i=1}^N(y_i(t)-\bar{y})^2 dt, \nonumber  
\end{equation}  
where $y_i(t)$ is the modified drift of each local clock and $\bar{y}$ is the average of the modified drift of the system.

The involved parameters are set to be $b_i=c_i=0.5/d_i^{out}$ for all $i\in\{1,\dots,N\}$. The same control law \eqref{eq:dynamic_drift} is applied. For \textbf{Algorithm 1} and \textbf{Algorithm 2} that achieve clock synchronization, their triggering functions are given by \eqref{eq:clk_Alg1_triggFunc} and \eqref{eq:clk_Alg2_triggFunc}, respectively. For the \textbf{Combined Algorithm}, its triggering function is given by \eqref{eq:clk_ComAlg_triggFunc}, with $\lambda=0.5$. For each~$\sigma$, we run $10$ simulations with random clock drift that satisfies to $\gamma_i\in (0.7,1.3)$ and obtain the average of $N_e$, $T_{con}$, $E$, and $C$ in each simulation.  

\begin{figure}
	\centering
	\subfloat[Network 1]{
		\label{fig:Ne_Tcon_rand}
		\resizebox*{6.5cm}{!}{\includegraphics{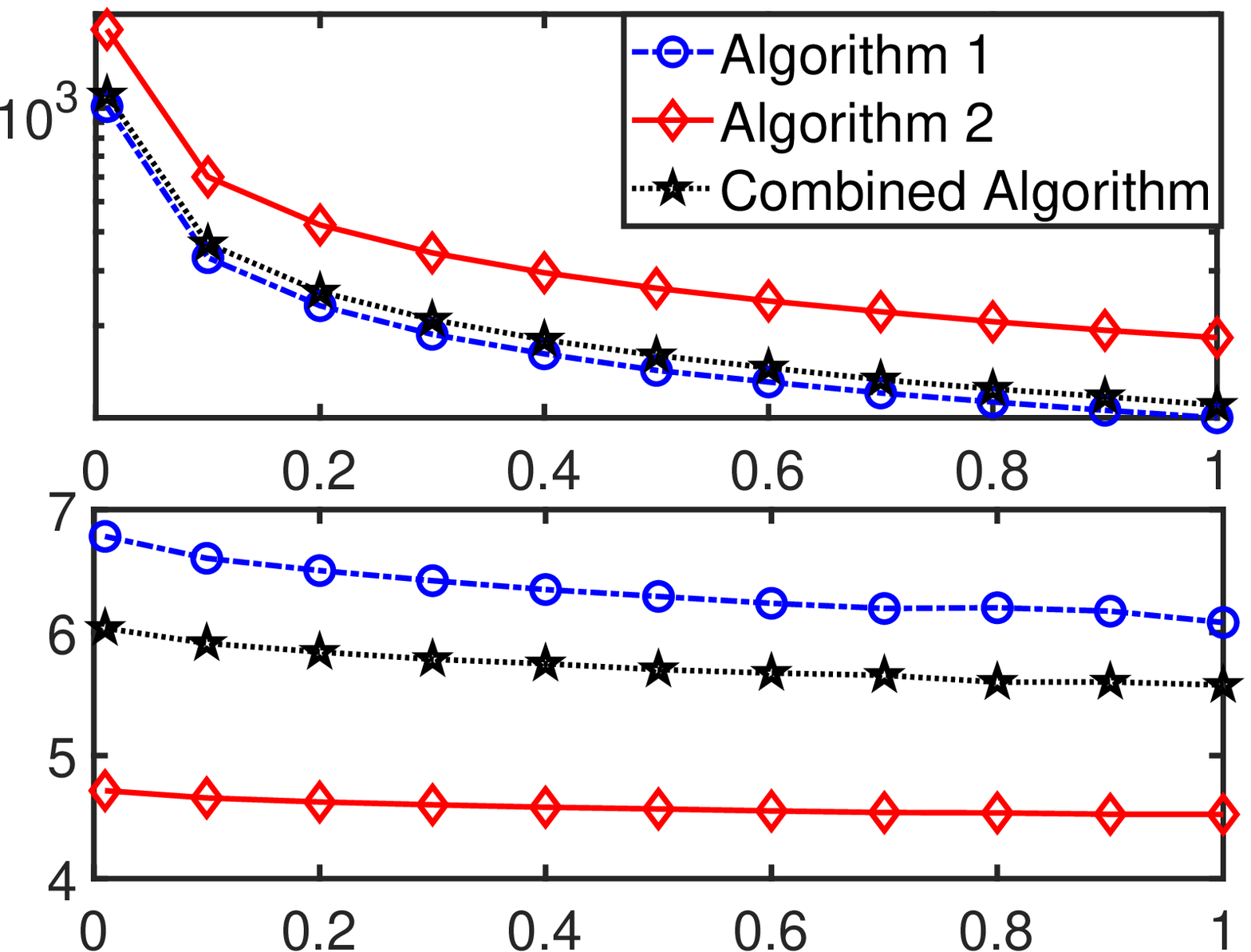}}
		\put(-190,105){\scriptsize $N_e$}
		\put(-192,38){\scriptsize $T_{con}$}
		\put(-90,-7){\scriptsize $\sigma$}}\hspace{20pt} 
	\subfloat[Network 1]{
		\label{fig:energy_C_rand}
		\resizebox*{6.5cm}{!}{\includegraphics{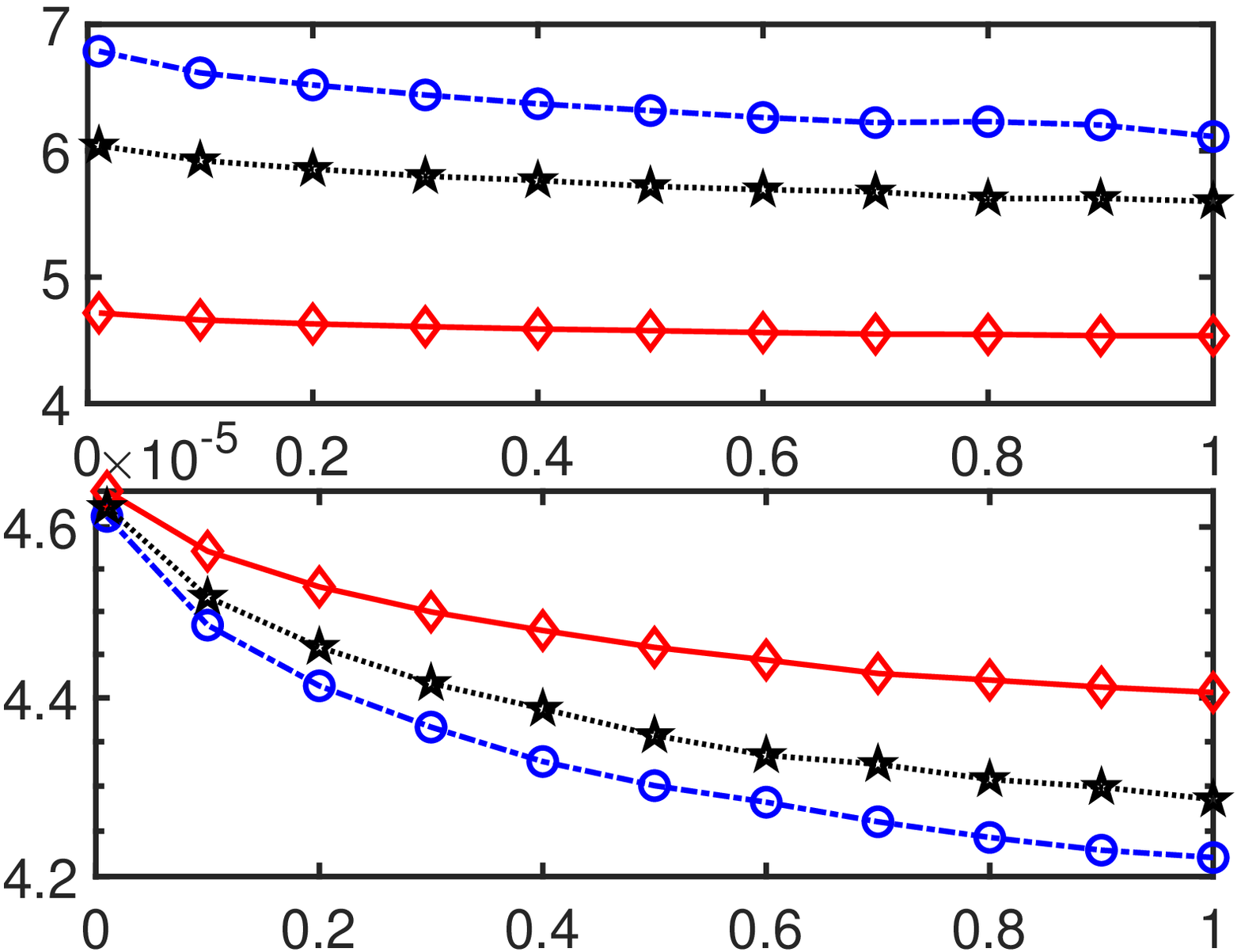}}
		\put(-187,105){\scriptsize $E$}
		\put(-191,38){\scriptsize $\mathcal{C}$}
		\put(-90,-7){\scriptsize $\sigma$}}\\
	\subfloat[Network 2]{
		\label{fig:Ne_Tcon_ring}
		\resizebox*{6.5cm}{!}{\includegraphics{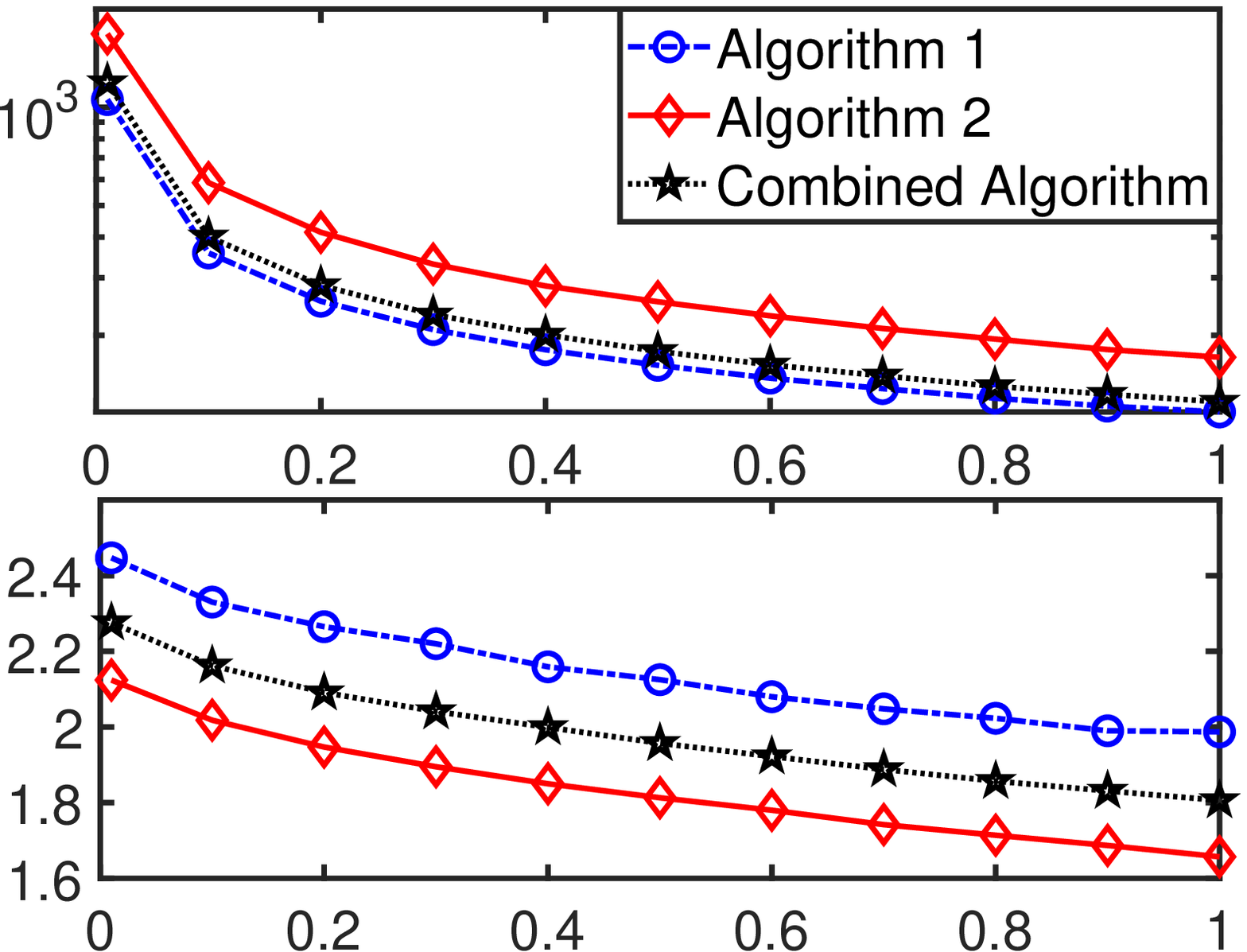}}
		\put(-190,105){\scriptsize $N_e$}
		\put(-200,38){\scriptsize $T_{con}$}
		\put(-90,-7){\scriptsize $\sigma$}}\hspace{20pt} 
	\subfloat[Network 2]{
		\label{fig:energy_C_ring}
		\resizebox*{6.5cm}{!}{\includegraphics{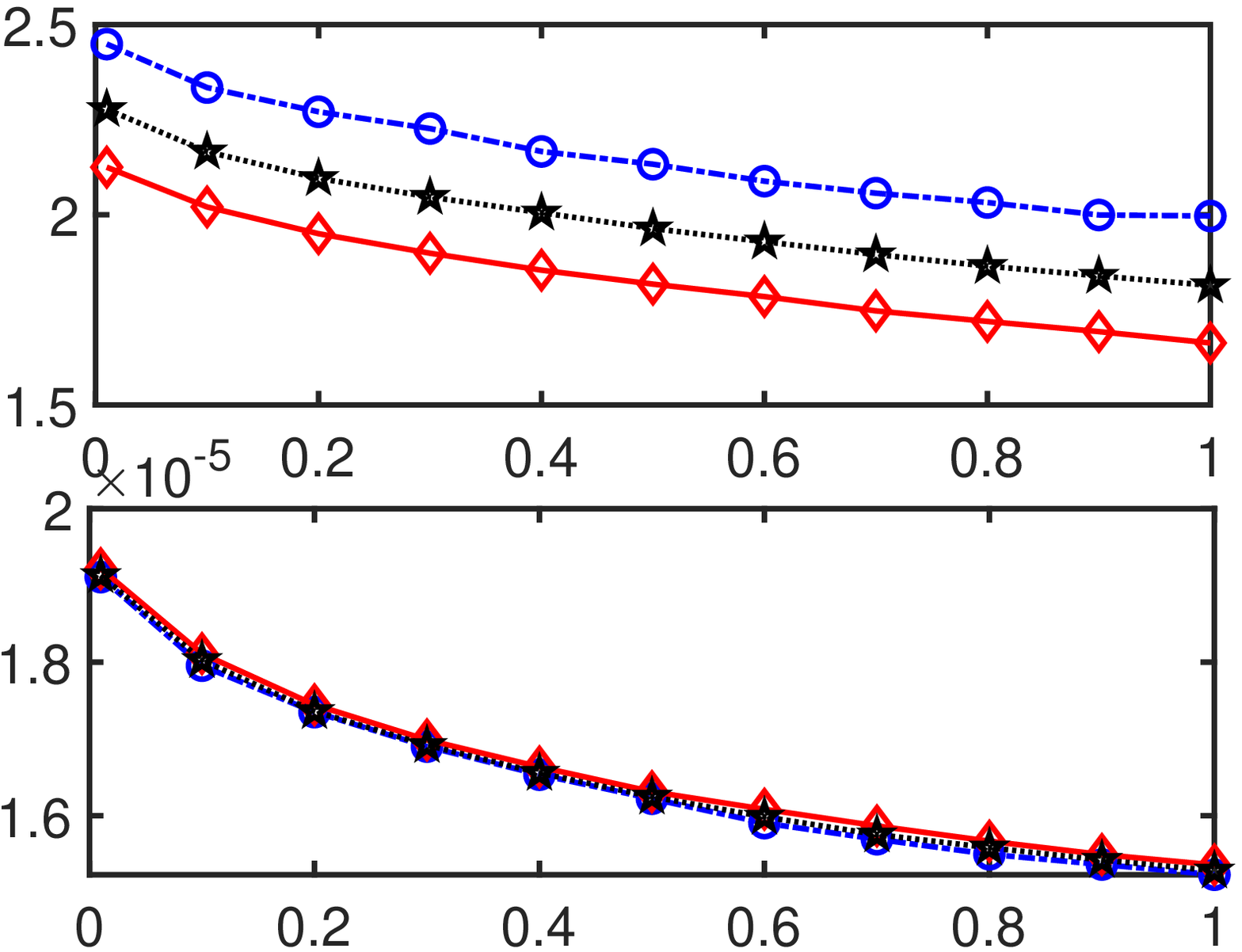}}
		\put(-187,105){\scriptsize $E$}
		\put(-191,35){\scriptsize $\mathcal{C}$}
		\put(-90,-7){\scriptsize $\sigma$}}\\
	\subfloat[Network 3]{
		\label{fig:Ne_Tcon_comp}
		\resizebox*{6.5cm}{!}{\includegraphics{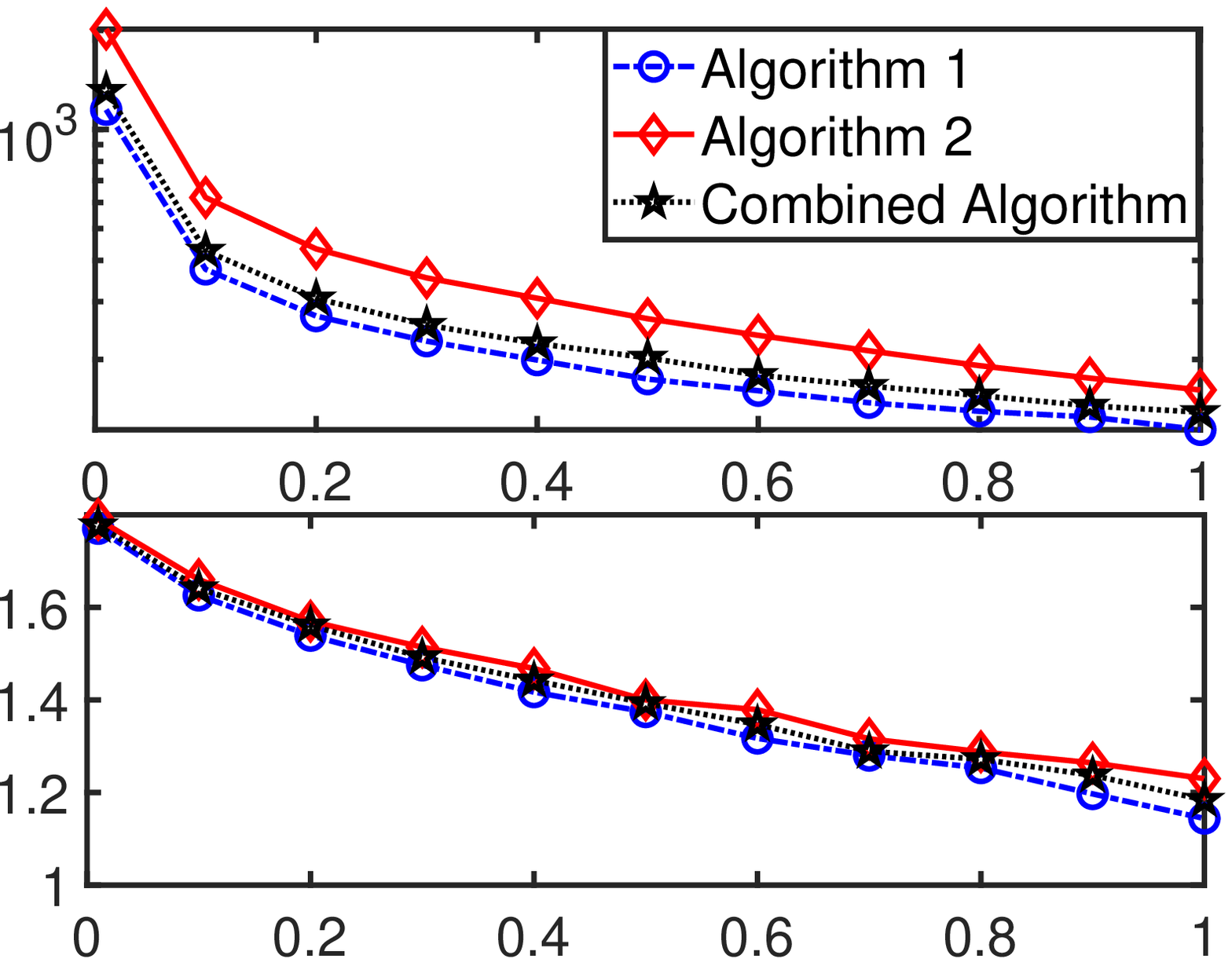}}
		\put(-190,105){\scriptsize $N_e$}
		\put(-202,38){\scriptsize $T_{con}$}
		\put(-90,-7){\scriptsize $\sigma$}}\hspace{20pt} 
	\subfloat[Network 3]{
		\label{fig:energy_C_comp}
		\resizebox*{6.5cm}{!}{\includegraphics{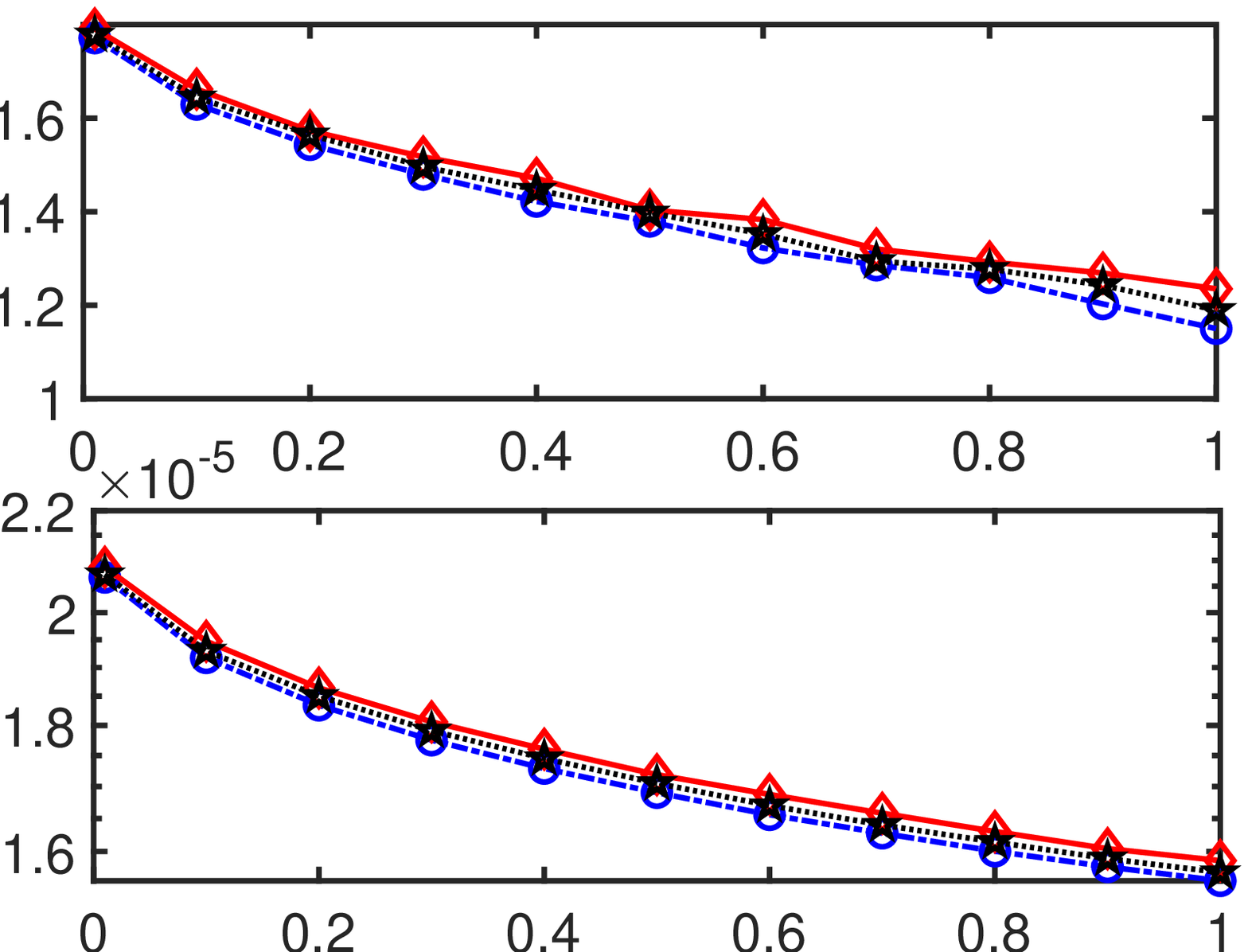}}
		\put(-194,105){\scriptsize $E$}
		\put(-191,38){\scriptsize $\mathcal{C}$}
		\put(-90,-7){\scriptsize $\sigma$}}\\
	\subfloat[Network 4]{
		\label{fig:Ne_Tcon_star}
		\resizebox*{6.5cm}{!}{\includegraphics{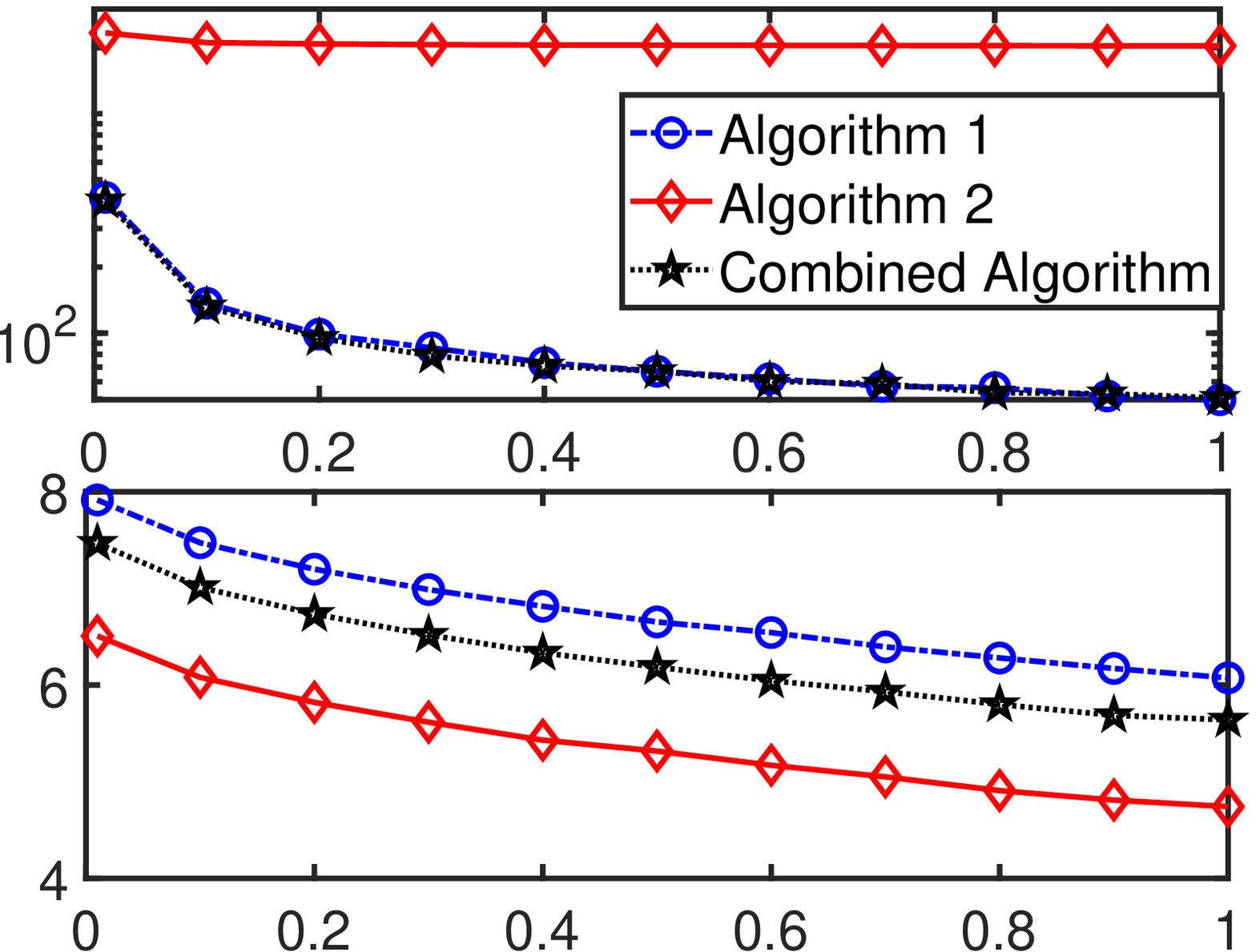}}
		\put(-190,105){\scriptsize $N_e$}
		\put(-199,38){\scriptsize $T_{con}$}
		\put(-90,-7){\scriptsize $\sigma$}}\hspace{20pt} 
	\subfloat[Network 4]{
		\label{fig:energy_C_star}
		\resizebox*{6.5cm}{!}{\includegraphics{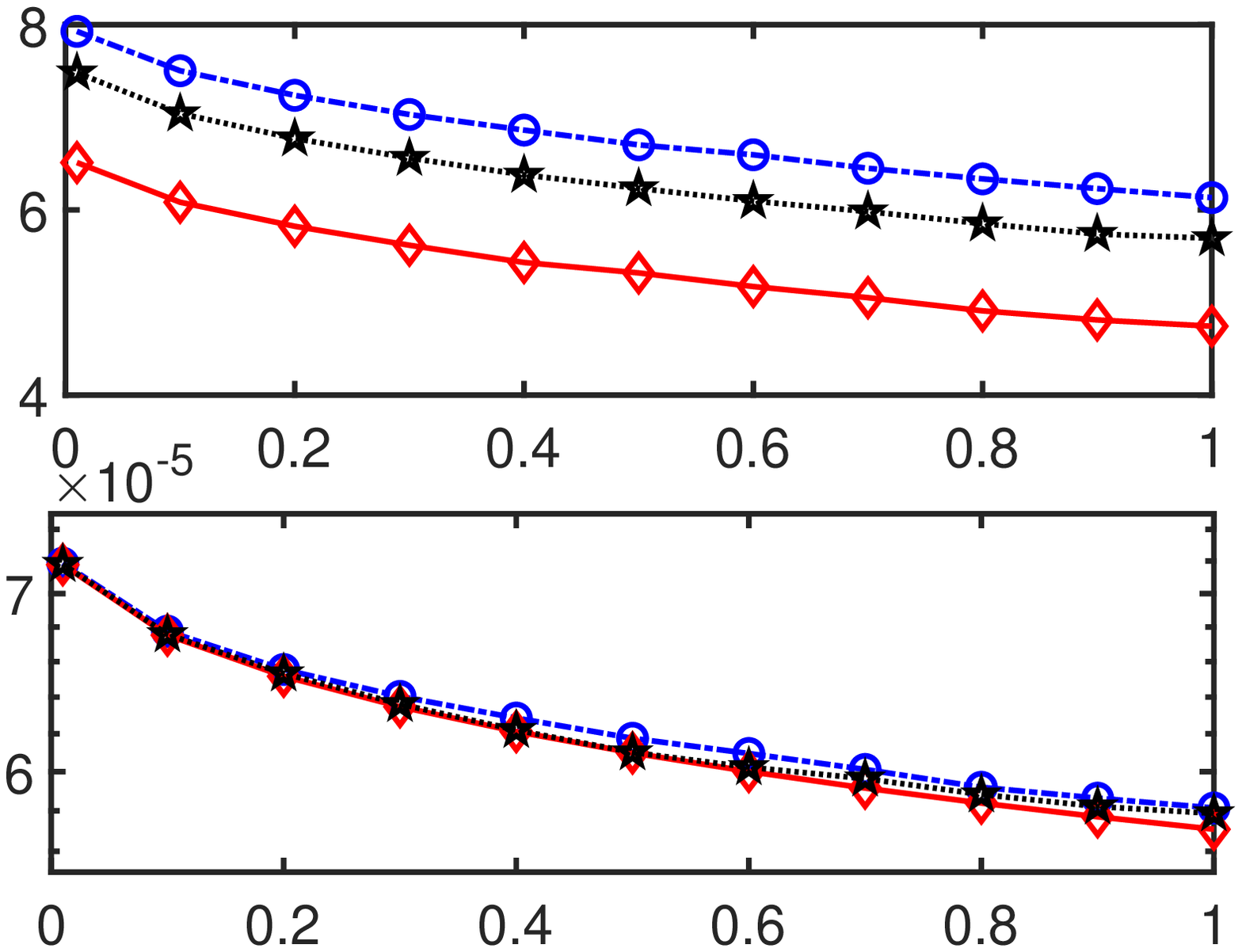}}
		\put(-192,105){\scriptsize $E$}
		\put(-191,38){\scriptsize $\mathcal{C}$}
		\put(-90,-7){\scriptsize $\sigma$}}
	\caption{Plots of different evaluation metrics. For figure (a), (c), (e), (g), top: total events triggered, bottom: convergence time (bottom); for figure (b), (d), (f), (h), top: energy consumption, bottom: $H_2$-norm squared. }
	\label{fig:performance}
\end{figure}

From the top figures in Figure \ref{fig:Ne_Tcon_rand}, \ref{fig:Ne_Tcon_ring}, \ref{fig:Ne_Tcon_comp}, and \ref{fig:Ne_Tcon_star}, we can see that for different $\sigma$, the total number of events triggered in the system when executing \textbf{Algorithm 2} is larger than that when executing \textbf{Algorithm 1}. On the other hand, from the bottom figures in Figure \ref{fig:Ne_Tcon_rand}, \ref{fig:Ne_Tcon_ring}, \ref{fig:Ne_Tcon_comp}, and \ref{fig:Ne_Tcon_star}, we can see that the time needed to reach a 99\% convergence of the system is usually much less when executing \textbf{Algorithm 2}, with the only exception for the complete network, where both algorithms have similar convergence speed. As the total communication energy consumption is related with both the total number of events triggered and the time required to reach convergence, we can see from the top figures in Figure \ref{fig:energy_C_rand}, \ref{fig:energy_C_ring}, \ref{fig:energy_C_comp}, and \ref{fig:energy_C_star} that either algorithm can outperform the other in terms of the total energy consumption for different network topologies. The $H_2$-norm squared evaluates the distance of each local modified drift with the average modified drift of the system, whose value therefore also indicates the convergence speed of the system to some extent, see the bottom figures in Figure \ref{fig:energy_C_rand}, \ref{fig:energy_C_ring}, \ref{fig:energy_C_comp}, and \ref{fig:energy_C_star}. Therefore, depending on different network topologies and depending on what performance metrics are most important for the application at hand, it may be desirable to implement different types of event-triggered algorithms. Note that the \textbf{Combined Algorithm} can easily be tuned to approach either \textbf{Algorithm 1} or \textbf{Algorithm 2} or anything in between to meet varying system needs by setting values for~$\lambda$. This also motivates our future work of adapting~$\lambda$ online to further improve performance.

\section{Conclusion}
\label{sec:con}
This paper proposes a class of distributed event-triggered communication and control law for multi-agent systems whose underlying directed graphs are weight-balanced. The class of algorithms are developed from a class of Lyapunov functions, each of which is a linear combination (parameterized by~$\lambda \in [0,1]$) of two Lyapunov functions. Each~$\lambda$ defines a new Lyapunov function coupled with a new event-triggered coordination algorithm which uses that particular function to guarantee correctness and is able to exclude the possibility of Zeno behavior. We show that the proposed entire class of event-triggered algorithms can be tuned to meet varying performance needs by adjusting $\lambda$. We also apply the proposed distributed event-triggered algorithms to solve the practical clock synchronization problem in WSNs. For the future research, we will focus on developing a unified evaluation metric (which is a function of different performance needs) that can be used to evaluate the performance of different algorithms. In that way, the class of distributed algorithms will be developed from a tunable algorithm to an adaptive algorithm.

 

\section*{Funding}
This work was supported by the NSF under Grant \#204294. 

\bibliographystyle{apacite}
\bibliography{reference}

\appendix
\section{Proof of Lemma \ref{lemma:V2_dot_upbound}}
\begin{proof}
	Omit the time stamp $t$ for simplicity. The derivative of $V_2(x)$ takes the form
	\begin{equation}
	\label{eq:V2_dot}
	\textstyle \dot{V}_2(x)=x^TL^T\dot{x}.
	\end{equation}
	Substitute the vector form $x=\hat{x}-e$ into \eqref{eq:V2_dot}, and expand it with \eqref{eq:single_dymic}, we have
	\begin{equation}
	\label{eq:V2_dot_expansion}
	\begin{split}
	\textstyle \dot{V}_2(x)&=\hat{x}^TL^T\dot{x}-e^TL^T\dot{x}\\
	&=\sum_{i=1}^N(\sum_{j\in \mathcal{N}_i^{out}}w_{ij}(\hat{x}_i-\hat{x}_j)u_i-\sum_{j\in \mathcal{N}_i^{out}}w_{ij}(e_i-e_j)u_i)\\
	&=\sum_{i=1}^N(-u_i^2-\sum_{j\in\mathcal{N}_i^{out}}w_{ij}e_iu_i+\sum_{j\in\mathcal{N}_i^{out}} w_{ij}e_ju_i)\\
	&=\sum_{i=1}^N(-u_i^2-d_i^{out}e_iu_i+\sum_{j\in \mathcal{N}_i^{out}} w_{ij}e_ju_i).
	\end{split}
	\end{equation}
	For $b_i,\;c_j>0$, applyYoung's inequality \eqref{eq:youngs} to the cross terms at the right hand side of \eqref{eq:V2_dot_expansion} gives
	\begin{equation}
	\textstyle \begin{split}
	-d_i^{out}e_iu_i &\leq \frac{d_i^{out}}{2b_i}e_i^2+\frac{d_i^{out}b_i}{2}u_i^2,\\
	\sum_{j\in \mathcal{N}_i^{out}} w_{ij}e_ju_i &\leq \sum_{j\in \mathcal{N}_i^{out}}\frac{w_{ij}}{2c_j}e_j^2+\sum_{j\in \mathcal{N}_i^{out}}\frac{w_{ij}c_j}{2}u_i^2.
	\end{split}
	\nonumber
	\end{equation}
	Since the digraph is weight-balanced, the following equality holds:
	\begin{equation}
	\textstyle \begin{split}
	\sum_{i=1}^N\sum_{j\in \mathcal{N}_i^{out}}\frac{w_{ij}}{2c_j}e_j^2
	=\sum_{i=1}^N\sum_{j\in \mathcal{N}_i^{in}}\frac{w_{ji}}{2c_i}e_i^2 
	=\sum_{i=1}^N \frac{d_i^{in}}{2c_i}e_i^2  
	=\sum_{i=1}^N \frac{d_i^{out}}{2c_i}e_i^2.
	\end{split}
	\nonumber
	\end{equation}
	Combine the above inequalities and equality, we obtain an upper bound for $\dot{V}_2(x)$:
	\begin{equation}
	\label{eq:V2_dot_upper}
	\textstyle \begin{split}
	\dot{V}_2(x) &\leq \sum_{i=1}^N\Big(-u_i^2+\frac{d_i^{out}e_i^2}{2b_i}+\frac{d_i^{out}b_iu_i^2}{2}+\frac{d_i^{out}e_i^2}{2c_i}+\sum_{j\in\mathcal{N}_i^{out}}\frac{w_{ij}c_j}{2}u_i^2\Big)\\
	&=-\sum_{i=1}^N \Bigg[\Big(1-\frac{d_i^{out}b_i}{2}-\sum_{j\in\mathcal{N}_i^{out}}\frac{w_{ij}c_j}{2}\Big)u_i^2-\Big(\frac{d_i^{out}}{2b_i}+\frac{d_i^{out}}{2c_i}\Big)e_i^2\Bigg]\\
	&=-\sum_{i=1}^N \Bigg[\delta_i u_i^2-\Big(\frac{d_i^{out}}{2b_i}+\frac{d_i^{out}}{2c_i}\Big)e_i^2\Bigg],
	\end{split}
	\end{equation}
	with $\delta_i$ defined in \eqref{eq:delta}. To ensure $\delta_i>0$, we require $b_i,c_j<\frac{1}{d_i^{out}}$.
\end{proof}

\end{document}